\documentclass[preprint,12pt,3p]{elsarticle}

\usepackage{amssymb}
\usepackage{amsthm}
\usepackage{amsmath}
\usepackage{xcolor,color,enumerate}
\usepackage[color,all]{xy}

\theoremstyle{plain}
\newtheorem{theorem}{Theorem}
\newtheorem{corollary}[theorem]{Corollary}
\newtheorem{proposition}[theorem]{Proposition}
\newtheorem{lemma}[theorem]{Lemma}
\newtheorem{example}{Example}

\theoremstyle{definition}
\newtheorem{definition}[theorem]{Definition}
\newtheorem{note}[theorem]{Note}

\newcommand{\partd}[1]{\frac{\partial}{\partial {#1}}}

\newcommand{\Hil}{\mathcal{H}}

\newcommand{\C}{\mathbb{C}}

\renewcommand{\P}{\mathcal{P}}

\renewcommand{\d}{{\rm d}}

\numberwithin{equation}{section}
\numberwithin{theorem}{section}
\numberwithin{example}{section}
\input xy
\xyoption{all}

\biboptions{sort&compress}

\journal{J. Differential Equations}

\begin{document}

\begin{frontmatter}



\title{Lie systems and Schr\"odinger equations}


\author{J.F. Cari\~nena}

\address{Faculty of Sciences and IUMA, University of Zaragoza \\
c. Pedro Cerbuna 12, 50.009 Zaragoza, Spain.}

\author{J. Clemente-Gallardo and J.A. Jover-Galtier}
\address{Instituto de Biocomputaci\'on y F\'isica de Sistemas Complejos (BIFI)\\
c. Mariano Esquillor (Edificio I+D), 50.018 Zaragoza, Spain.}

\author{J. de Lucas}
\address{Department of Mathematical Methods in Physics, University of Warsaw, \\ ul. Pasteura 5, 02-093, Warszawa, Poland}

\begin{abstract}
We prove that $t$-dependent Schr\"odinger equations on finite-dimensional Hilbert spaces determined by $t$-dependent Hermitian Hamiltonian operators 
 can be described through Lie systems admitting a Vessiot--Guldberg Lie algebra of K\"ahler vector fields. This result is extended to other related Schr\"odinger equations, e.g. projective ones, and their properties are studied through Poisson, presymplectic and K\"ahler structures.
This leads to derive nonlinear superposition rules for them depending in a lower (or equal) number of solutions than standard linear ones. 
Special attention is paid to applications in  $n$-qubit systems.

\end{abstract}

%
%


\end{frontmatter}

\section{Introduction}
It is undoubtable that geometric techniques, e.g. Lie symmetries or jet bundles, have become a standard tool in the study of differential equations and related problems  \cite{CRC94,Be84,Olver}.
In particular, this work focuses on the geometric analysis of Lie systems appearing in quantum mechanics
\cite{CGM00,CGM07,Ib99,LS,PW}. A {\it Lie system} is a non-autonomous system of first-order ordinary differential
equations  whose general solution can be written in terms of a generic finite
family of particular solutions and a set of constants via a (generally) nonlinear function, a so-called {\it superposition rule} \cite{CGM00,CGM07,LS,PW}.

Lie systems occur in the research on the integrability of quantum systems \cite{ADR12}, $t$-dependent
Schr\"odinger equations \cite{CLR09},  $t$-dependent frequency Smorodinsky--Winternitz oscillators \cite{BCHLS13}, several types of Ermakov systems and Milne--Pinney equations \cite{CL,CS15,CLR08}, string theory \cite{CV14}, deformation of mechanical systems \cite{Fi14},  control systems \cite{MT14}, etcetera (see \cite{Dissertations}). 

Lie systems of physical or mathematical relevance can be studied via symplectic \cite{ADR12,CGM00}, Poisson \cite{CLS13}, $k$-symplectic \cite{LV15}, Jacobi \cite{HLS15AIMS}, Dirac \cite{CGLS13} and Nambu structures \cite{LS16}. This allows one to use geometric techniques to analyse their properties, e.g. their solutions \cite{ADR12}, constants of motion \cite{BCHLS13}, Lie symmetries \cite{CLS13}, and other features \cite{CGM00,CGM07,FLV10,Ib99,PW}. This has also led to develop new mathematical tools so as to investigate Lie systems \cite{LV15}.

Although Lie systems have already been applied in quantum mechanical systems \cite{BHLS15,CLR09,CR03,CR04}, there still exist many open problems. In particular, this article addresses the application of Lie systems in $t$-dependent Schr\"odinger equations on finite-dimensional Hilbert spaces and their projections/restrictions to relevant spaces, e.g. projective Schr\"odinger equations \cite{BH08}. A special role is played by the use of geometric structures, e.g. K\"ahler structures, which enables us to calculate superposition rules through the distributional approach devised in \cite{CGM07} and its refinement for Lie systems with compatible geometric structures \cite{BCHLS13}.
It is worth noting that K\"ahler structures naturally appear as a consequence of the quantum nature of the problems under study.

In geometric terms, the {\it Lie--Scheffers Theorem} \cite{LS,Dissertations,CGM07} states that a Lie
system amounts to a $t$-dependent vector field taking
values in a finite-dimensional Lie algebra of vector fields: a {\it
	Vessiot--Guldberg Lie algebra} (VG--Lie algebra) \cite{Dissertations,Ib00,Ib09}.
	
A particular branch of the research on Lie systems is devoted to the study of Lie
systems admitting a VG--Lie algebra
of Hamiltonian vector fields with respect to a geometric
structure.
In the pioneering work \cite{CGM00},
the authors briefly analysed Lie systems with a VG--Lie algebra of
Hamiltonian vector fields with respect to a symplectic structure.
The study of Lie--Hamilton systems, i.e. Lie systems with a
VG--Lie algebra of Hamiltonian vector fields relative to a Poisson structure,
was initiated in \cite{CLS13}. This gave rise to  new methods to
investigate such Lie systems
\cite{Ru08,Ru10,BBHLS13}. 
Lie systems admitting a VG--Lie algebra of Hamiltonian
vector fields relative to Dirac and $k$-symplectic
structures were studied in \cite{LV15,CGLS13}.  Recently, Lie systems with a VG--Lie algebra of Hamiltonian
vector fields relative to a Nambu structure have been investigated in \cite{LS16}.

The first aim of this work is to show that $t$-dependent Schr\"odinger
equations on a finite-dimensional Hilbert space $\mathcal{H}:=\mathbb{C}^n$ related to a $t$-dependent Hermitian Hamiltonian operator can be
studied through Lie systems admitting a VG--Lie algebra $V_{\mathcal{M}_{2n}} \simeq\mathfrak{u}(n)$
of K\"ahler vector fields  with respect to the K\"ahler structure
induced by the natural Hermitian product on $\mathbb{C}^n$ \cite{BH01,CCM07}. 
We prove that $t$-dependent Schr\"odinger equations on $\mathbb{C}^n$
related to $t$-dependent traceless Hermitian Hamiltonian operators 
admit 
nonlinear superposition rules depending on $n-1$ particular solutions. Thus, such quantum systems can be endowed with 
a simple,  generally nonlinear, superposition rule  allowing us to recover their general solutions by means of a lower number of particular solutions than by standard linear superposition rules.

The Lie groups $U(1)$ and $\mathbb{R}_+$ act freely on $\mathbb{C}_0^n:=\mathbb{C}^n\backslash\{(0,\ldots,0)\}$, by multiplication. The corresponding spaces of orbits are denoted by $\mathbb{C}_0^n/U(1)$ and $\mathbb{C}^n_0/\mathbb{R}_+$.  To highlight that previous spaces can be considered as real  manifolds,  they will be denoted  $\mathcal{R}_n$ and $\mathcal{S}_n$, respectively. Likewise, the spaces $\mathbb{C}^n$, $\mathbb{C}_0^n$ and $\mathbb{C}_0^n/\mathbb{C}_0$ will be represented by $\mathcal{M}_{2n}$, $\mathcal{M}^\times_{2n}$ and $\mathcal{P}_n$, respectively.
We prove that the restriction to $\mathcal{M}_{2n}^\times$ of the $t$-dependent Schr\"odinger
equations referred in the preceding paragraph can be projected onto $\mathcal{R}_n$ and restricted to the unit sphere $\mathcal{S}_n$ giving rise to Lie systems admitting VG--Lie algebras of Hamiltonian vector fields relative to different geometric structures, e.g. Dirac and Poisson structures. 

Subsequently, the solutions of the
referred to as {\it projective Schr\"odinger equations} \cite{BH01} are recovered
through the projection onto the projective space
$\mathcal{P}_{n}$ of the $t$-dependent 
Schr\"odinger equations on $\mathcal{M}^\times_{2n}$. This allows us to understand geometrically standard projective Schr\"odin-ger
equations as Lie systems admitting a
VG--Lie algebra of K\"ahler vector fields relative to the K\"ahler structure induced by the Study--Fubini metric on
$\mathcal{P}_{n}$ \cite{BH01}.

Above findings suggest us to define a new type of Lie systems
possessing a Lie algebra of K\"ahler vector fields with respect to
a K\"ahler structure, the {\it K\"ahler--Lie systems}, and to use techniques from Riemannian and symplectic
geometry to study them. 


Using our results we derive geometrically superposition rules for $t$-dependent Schr\"odinger equations on $\mathcal{M}_{2n}$ related to $t$-dependent traceless Hermitian Hamiltonian operators. This allows us to obtain  superposition rules without the integration of vector fields or PDEs as in standard methods \cite{CGM07,PW}. Similarly, we study and calculate superposition rules for the projections of the previous Schr\"odinger equations on certain spaces $\mathcal{S}_n$, $\mathcal{R}_n$ and $\mathcal{P}_{n}$. Schr\"odinger equations on $\mathcal{M}_4,\mathcal{R}_2,\mathcal{P}_2$ are analysed in detail.  Most relevant results concerning the preceeding equations and superposition rules are summarised in Table \ref{table1}. Their interest is due to its occurrence in the research on qubits.

\!\!\!\!\!
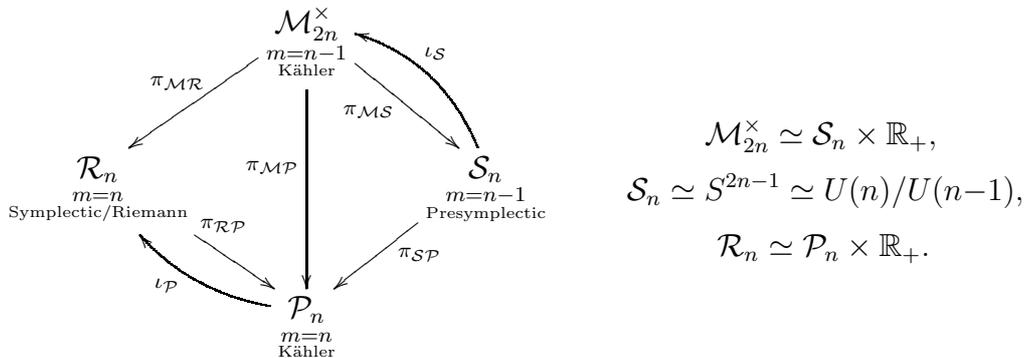
\begin{table}[h!]
	\label{table1}
	\caption{{\small {The following diagram illustrates the geometric structures and natural inclusions employed to study the Lie systems induced by the projection on each space of $t$-dependent Schr\"odinger equations on $\mathcal{M}_{2n}^\times$ related to traceless Hermitian Hamiltonian operators. The number $m$ stands for the number of particular solutions of their superposition rules. The right  column shows some known diffeomorphisms used in our work.
			}}}
			\begin{center}
%
\begin{minipage}{8cm}
					\xymatrix@C=2em@R=1.8em{
						&\underset{\underset{\text{\rm K\"ahler}}{m=n-1}}{\mathcal{M}_{2n}^\times}
						\ar[dd]_{\pi_{\mathcal{MP}}}
						\ar[dr]_{\pi_{\mathcal{MS}}}
						\ar[ld]_{\pi_{\mathcal{MR}}}&\\
						\underset{\underset{\text{Symplectic/Riemann}}{m=n}}{\mathcal{R}_n}
						\ar[dr]^{\pi_{\mathcal{RP}}}&&
						\underset{\underset{\text{Presymplectic}}{m=n-1}}{\mathcal{S}_n}
						\ar@/_{5mm}/[ul]_{\iota_{\mathcal{S}}}
						\ar[dl]^{\pi_{\mathcal{SP}}}\\
						&\underset{\underset{\text{\rm K\"ahler}}{m=n}}{\mathcal{P}_{n}}
						\ar@/^{5mm}/[ul]^{\iota_{\mathcal{P}}}&}
				\end{minipage}
				\begin{minipage}{6cm}
					$$
					\qquad\mathcal{M}_{2n}^\times \simeq \mathcal{S}_n \times \mathbb{R}_+,\\
					$$
					$$
					\qquad\mathcal{S}_n \simeq S^{2n-1}\simeq U(n)/U(n-1),\\
					$$
					$$
					\qquad\mathcal{R}_n \simeq \mathcal{P}_{n}\times \mathbb{R}_+. \\
					$$
				\end{minipage}
			\end{center}
		\end{table}\!\!\!\!\!\!\!\!\!\!
		The structure of the paper goes as follows. Section 2 concerns the description
		of the notions and conventions about $t$-dependent vector
		fields, Lie algebras, Lie systems and geometric structures to be used hereafter. Section 3 focuses  on the formalism of geometric quantum mechanics.
		Section 4 addresses the description of $t$-dependent Schr\"odinger equations on an $n$-dimensional Hilbert space through a Lie system admitting a VG--Lie algebra of K\"ahler vector fields on $\mathbb{C}^n$. The projection of  $t$-dependent Schr\"odinger equations onto $\mathcal{R}_n$ and their restrictions on $\mathcal{S}_n$ are analysed in Section 5. A Lie systems approach to projective $t$-dependent Schr\"odinger equations is given in Section 6.  In Section 7 we study superposition rules for $t$-dependent Schr\"odinger equations and their projections to previous spaces.  Section 8 is devoted to provide superposition rules for one-qubit systems and their projections onto the above mentioned spaces. The cases of $n$-qubit systems and other $t$-dependent Schr\"odinger equations and their projections are analysed in Section 9. Our results and future work are summarised in Section 10.

		\section{Fundamentals}\label{LSLS}
		
		If not otherwise stated, we assume mathematical objects to be real, smooth,
		and globally defined to omit 
		minor technical problems and to highlight main results. Systems of differential equations are assumed to be non-autonomous systems of ordinary differential equations.
		
		Let $(V,[\cdot,\cdot])$ be a Lie algebra with Lie bracket $[\cdot\,,\cdot]:V\times V\rightarrow V$. For the sake of simplicity, we will denote the Lie algebra by $V$ if $[\cdot,\cdot]$ is known from context. Given subsets
		$\mathcal{A}, \mathcal{B}\subset V$, we write
		$[\mathcal{A},\mathcal{B}]$ for the linear subspace of $V$ spanned by the
		Lie brackets between elements
		of $\mathcal{A}$ and $\mathcal{B}$, and we define ${\rm
			Lie}(\mathcal{B},[\cdot,\cdot])$ to be  the smallest Lie subalgebra
		of $V$ containing $\mathcal{B}$. We will simply write ${\rm Lie}(\mathcal{B})$ if it is clear what we mean.
		
		A {\it generalised distribution} $\mathcal{D}$ on a manifold $N$ is a function mapping each $x\in N$ to  a linear
		subspace $\mathcal{D}_x\subset T_xN$. We say that $\mathcal{D}$ is 
		{\it regular at} $x'\in N$ if ${\rm r}:x\in N\mapsto \dim\mathcal{D}_x\in \mathbb{N}\cup\{0\}$  is locally constant
		around $x'$. Similarly, $\mathcal{D}$ is said to be {\it regular on an open $U\subset
			N$} when ${\rm r}$ is constant on $U$. Finally,
		a vector field  $Y$ on $N$ takes values in $\mathcal{D}$, in short $Y\in\mathcal{D}$, if $Y_x\in\mathcal{D}_x$ for all $x\in N$.

		A {\it $t$-dependent vector field} $X$ on $N$ is a map $X:(t,x)\in\mathbb{R}\times
		N\mapsto X(t,x)\in TN$
		such that $\tau_N\circ X=\pi_2$, where $\pi_2:(t,x)\in\mathbb{R}\times N\mapsto
		x\in N$ and $\tau_N$ is the canonical projection of the tangent bundle on $N$
		. A $t$-dependent vector field $X$ on $N$ amounts to
		a family of vector fields $\{X_t\}_{t\in\mathbb{R}}$ on  $N$, where $X_t:x\in N\mapsto
		X(t,x)\in TN$ for all $t\in\mathbb{R}$ \cite{Dissertations}. A $t$-dependent vector field $X$ is {\it projectable} relative to a map $\pi:N\rightarrow M$  when
		$X_t$ is projectable with respect to $\pi$ for each $t\in \mathbb{R}$.
		
		The  {\it smallest Lie algebra} of $X$ is the smallest real Lie subalgebra,
		$V^X$, containing   $\{X_t\}_{t\in\mathbb{R}}$, namely $V^X={\rm
			Lie}(\{X_t\}_{t\in\mathbb{R}})$. Every Lie algebra $V$ of vector fields on $N$ induces an integrable generalised distribution  $\mathcal{D}^V:=\{X(x)|X\in V,x\in N\}\subset TN$ on $N$.
		
		An {\it integral curve} of $X$ is an integral curve
		$\gamma:\mathbb{R}\mapsto \mathbb{R}\times N$ of the {\it suspension} of $X$,
		i.e. the vector field $X(t,x)+\partial/\partial t$ on $\mathbb{R}\times N$ \cite{FM}. The curve $\gamma$ always admits a reparametrisation $\bar
		t=\bar t(t)$ such that
		$$
		\frac{{\rm d}(\pi_2 \circ \gamma)}{{\rm d}\bar t}(\bar t)=(X\circ \gamma)(\bar t).
		$$
		This system is referred to as the {\it associated system} of $X$. Conversely,
		a system of first-order differential equations in normal form is always the associated system of a unique $t$-dependent vector field. This induces  a
		bijection between $t$-dependent vector fields and systems of first-order 
		differential equations in normal form. This justifies to denote by $X$ both a $t$-dependent vector field and its
		associated system.

		\begin{definition} A {\it superposition rule} depending on $m$ particular
			solutions for a nonautonomous system $X$ on $N$
			is a map $\Phi: (u_{(1)}, \ldots,u_{(m)};k) \in N^{m}\times N \mapsto \Phi(u_{(1)}, \ldots,u_{(m)};k)\in N$ such that the general
			solution, $x(t)$, of $X$ can be written as
			$x(t)=\Phi(x_{(1)}(t), \ldots,x_{(m)}(t);k),$
			where $x_{(1)}(t),\ldots,x_{(m)}(t)$ is a generic set of
			particular solutions to $X$, and $k\in N$.
		\end{definition}
		
		\begin{example} It is known that a Riccati equation, namely
			\begin{equation}\label{Ricc}
			\frac{{\rm d}x}{{\rm d}t}=a_0(t)+a_1(t)x+a_2(t)x^2,\qquad x\in \mathbb{R},
			\end{equation}
			where $a_0(t),a_1(t),a_2(t)$ are $t$-dependent real functions satisfying $a_0(t)a_2(t)\neq 0$, is such that its general solution can be brought into the form $x(t)=\Phi(x_{(1)}(t),x_{(2)}(t),x_{(3)}(t); k)$, with $\Phi: \mathbb{R}^3\times\mathbb{R} \rightarrow \mathbb{R}$ defined by
			$$ 
			\Phi (u_{(1)},u_{(2)},u_{(3)};k) :=\frac{u_{(1)}(u_{(3)}-u_{(2)})+ku_{(2)}(u_{(1)}-u_{(3)})}{u_{(3)}-u_{(2)}+k(u_{(1)}-u_{(3)})},
			$$
			where $x_{(1)}(t),x_{(2)}(t),x_{(3)}(t)$ are different particular solutions to (\ref{Ricc}).
		\end{example}

		\begin{theorem} {\bf (The Lie--Scheffers Theorem  \cite{CGM07,LS})} A system $X$ on $N$ admits a superposition rule if and only if $X={{\sum_{\alpha=1}^r}}b_\alpha(t)X_\alpha$
			for a family $b_1(t),\ldots,b_r(t)$  of $t$-dependent functions and a
			basis $X_1,\ldots,X_r$ of a real Lie algebra of vector fields on $N$.
		\end{theorem}
		If $X$ possesses a superposition rule, then $X$ is called a {\it Lie system}. The associated real Lie algebra of vector fields
		$\langle X_1,\ldots,X_r\rangle$ is called a {\it VG--Lie algebra} of $X$. The Lie--Scheffers theorem amounts to saying that $X$ is a Lie system if and only if $V^X$ is finite-dimensional. This fact is the keystone of the theory of Lie systems. When $V^X$ consists of Hamiltonian vector fields relative to some geometric structure, much more powerful methods can be used to study Lie systems \cite{Ru08,Ru10,BBHLS13,LV15,CGLS13,LS16}.
		
		\begin{definition} A system $X$ on $N$ is a {\it Lie--Hamilton system} if
			$V^X$ is a VG--Lie algebra of Hamiltonian vector fields relative to some Poisson bivector field on $N$.
		\end{definition}
		\begin{note} A vector field $X$ is {\it Hamiltonian} relative to a Poisson bivector $\Lambda$ with Hamiltonian function $h$ if $X=-\widehat \Lambda ({\rm d}h)$ for $\widehat \Lambda:T^*N\to TN$,  given by $\widehat \Lambda:\theta \in T^*N\mapsto \Lambda(\theta,\cdot)\in TN$. This is the standard convention in geometric mechanics, while the definition $X=\widehat \Lambda({\rm d}h)$ is usually chosen in Poisson geometry \cite{IV}.
		
		\end{note}
		\begin{definition} A {\it Lie--Hamiltonian structure} is a triple
			$(N,\Lambda,h)$, where $\Lambda$ is a Poisson bivector on $N$ and  $h:
			(t,x)\in \mathbb{R}\times N\mapsto h_t(x):= h(t,x)\in  \mathbb{R}$ is such that
			${\rm Lie}(\{h_t\}_{t\in\mathbb{R}},\{\cdot,\cdot\}_\Lambda)$, where  $\{\cdot,\cdot\}_\Lambda$ is the Lie bracket induced by $\Lambda$ \cite{IV}, is
			finite-dimensional.
		\end{definition}
		
		\begin{theorem}\label{FunHam} {\bf (Characterisation of Lie--Hamilton systems \cite{CLS13})} A system $X$ on $N$ is a {\it Lie--Hamilton system} if and only
			if there exists a Lie--Hamiltonian structure $(N,\Lambda,h)$ such that $X_t$ is a Hamiltonian vector field for the function
			$h_t$ for each $t\in\mathbb{R}$. We say that ${\rm Lie}(\{h_t\}_{t\in\mathbb{R}},\{\cdot,\cdot\}_\Lambda)$ is a {\it Lie--Hamilton} algebra of $X$.
		\end{theorem}
		
		Lie--Hamilton algebras can be employed to find superposition rules and constants of motion for Lie--Hamilton systems in a more easy way than by standard methods \cite{BCHLS13}.
		
		\begin{example} A complex Riccati equation with $t$-dependent real coefficients \cite{BHLS15,CSR13} can be brought into 
			\begin{equation}
			\label{CRE}
			\left\{\begin{aligned}
			\frac{{\rm d}x}{{\rm d}t}&=a_1(t)+a_2(t)x+a_3(t)(x^2-y^2),\\
			\frac{{\rm d}y}{{\rm d}t}&=a_2(t)y+a_3(t)2xy,
			\end{aligned}\right.\qquad (x,y)\in \mathbb{R}^2,
			\end{equation}
			where $a_1(t),a_2(t)$ and $a_3(t)$ are arbitrary $t$-dependent real functions. Let us prove that (\ref{CRE}) is a Lie--Hamilton system on $\mathbb{R}^2_{y\neq 0}:= \mathbb{R}^2\backslash\{(x,0)\,|\,x\in\mathbb{R}\}$. The system (\ref{CRE}) is associated with the $t$-dependent vector field $X=\sum_{\alpha=1}^3a_\alpha(t)X_\alpha$, where
			$$
			X_1:=\frac{\partial }{\partial x},\qquad X_2:=x\frac{\partial }{\partial x}+y\frac{\partial }{\partial y},\qquad
			X_3:=(x^2-y^2)\frac{\partial }{\partial x}+2xy\frac{\partial }{\partial y}.
			$$
			Hence, $X$ takes values in the Lie algebra $V:=\langle X_1,X_2,X_3\rangle$. The
			 vector fields $X_1,X_2,X_3$ are Hamiltonian relative to the Poisson bivector $
			\Lambda:=y^{2}{\partial/\partial x \wedge \partial/\partial y}
			$ on $\mathbb{R}^2_{y\neq 0}$			
			with Hamiltonian functions
			\begin{equation}
			h_1=-\frac 1y,\qquad h_2=-\frac xy,\qquad h_3=-\frac{x^2+y^2}{y},
			\label{ab}
			\end{equation}
			respectively. 
			If $\{\cdot,\cdot\}_\Lambda:C^\infty(\mathbb{R}^2_{y\neq 0})\times C^\infty(\mathbb{R}^2_{y\neq 0})\rightarrow C^\infty(\mathbb{R}^2_{
				y\neq 0})$ stands for the Poisson bracket induced by $\Lambda$ (see \cite{IV}), then
			\begin{equation}
			\label{sl2Rh}
			\{h_1,h_2\}_\Lambda=-h_1,\qquad \{h_1,h_3\}_\Lambda=-2h_2,\qquad \{h_2,h_3\}_\Lambda=-h_3.
			\end{equation}
			Thus, $\left(\mathbb{R}^2_{y\neq 0},\Lambda,h:=a_0(t)h_1+a_1(t)h_2+a_2(t)h_3\right)$ is a Lie--Hamiltonian structure for $X$.
			If $V^X\simeq \mathfrak{sl}(2)$, then $(\mathcal{H}_\Lambda,\{\cdot,\cdot\}_\Lambda):= (\langle h_1,h_2,h_3\rangle,\{\cdot,\cdot\}_\Lambda)$ is a Lie--Hamilton algebra for $X$ isomorphic to $\mathfrak{sl}(2)$.
		\end{example}

		One of the main properties of Lie systems is the existence of  superposition rules. There are several methods to obtain them \cite{CGM07,Dissertations,PW}. We here choose a procedure that can be improved by  using Lie--Hamilton structures. This method is based upon the so-called {diagonal prolongations} of the vector fields \cite{CGM07}. Given a vector field $X$ on $N$ with local coordinate expression
		$$
		X=\sum_{j=1}^n X_j(x)\frac{\partial}{\partial x_j},\qquad x\in N,\quad n:=\dim N,
		$$
		its {\it diagonal prolongation} to $N^m:=N\times\ldots\times N$ ($m$-times) is the vector field on $N^m$:
		$$
		X^{[m]}(x^{(0)},\ldots,x^{(m-1)}):=\sum_{a=0}^{m-1} \sum_{j=1}^nX_j (x^{(a)}) \frac{\partial}{\partial x_j^{(a)}} = \sum_{a = 0}^{m-1} X^{(a)}(x^{(a)}), \quad (x^{(0)},\ldots,x^{(m-1)})\in N^m,
		$$ 
		where $x_j^{(a)}(x^{(0)},\ldots,x^{(m-1)}):=x_j (x^{(a)})$ for $j\in \overline{1,n}$ and $X^{(a)}(x^{(a)})=X(x^{(a)})$ stands for $X$ on the $a$-th copy of $N$ within $N^m$.
		
		To calculate a superposition rule for a Lie system on $N$ with a VG--Lie algebra $V$, we find the smallest $m\in 
		\mathbb{N}$ so that the diagonal prolongations of elements of $V$ to $N^m$ span an integral distribution of rank $\dim V$ at a generic point. Then, $m$ becomes the number of particular solutions involved in the superposition rule. The superposition rule can be obtained by deriving $n$ common first-integrals $I_1,\ldots,I_n$ for the diagonal prolongations to $N^{m+1}$ of the vector fields of $V$ satisfying that $\det(\partial(I_1,\ldots,I_n)/\partial (x_1^{(0)},\ldots,x_n^{(0)})) \neq 0$. This gives the superposition rule by assuming $I_1=k_1,\ldots,I_n=k_n$ and writing $x_1^{(0)},\ldots,x_n^{(0)}$ in terms of the remaining variables $x_i^{(a)}$, with $1<a\leq m$   and $k_1,\ldots,k_n$ (see \cite{CGM07,Dissertations} for details and examples).
		
		When a Lie system admits a VG--Lie algebra of Hamiltonian vector fields relative to some geometric structure, e.g. a symplectic or K\"ahler structure, there exist geometric and algebraic methods to obtain $I_1,\ldots,I_n$ and to simplify the description of the superposition rule \cite{BCHLS13}. This requires to prolongate geometric structures according to the following construction. Let $(E,N,\tau:E\to N)$ be a vector bundle. Its {\it diagonal prolongation} to  $N^m$ is a vector bundle $(E^{[m]},N^m,\tau^{[m]}:E^{[m]}\rightarrow N^m)$, where  $E^{[m]}:=E\times\cdots\times E$ ($m$-times) and $\tau^{[m]}$ is the only map satisfying that $\pi_{N,j}\circ \tau^{[m]}=\tau\circ \pi_{E,j}$ for $j=\overline{1,m}$, where $\pi_{E,j}:E^{[m]}\rightarrow E$ and $\pi_{N,j}:N^m\rightarrow N$ are the natural projections of $E^{[m]}$ and $N^m$ onto the $j$-th copy of $E$ and $N$ within $E^{[m]}$ and $N^m$, respectively.
		Every section $e:N\to E$ of $(E,N,\tau)$ has a natural {\it diagonal prolongation} to a section $e^{[m]}$ of $(E^{[m]},N^m,\tau^{[m]})$:
		\begin{equation*}
		e^{[m]}(x^{(0)},\dots,x^{(m-1)}):=e(x^{(0)})+\cdots +e(x^{(m-1)})\,.
		\end{equation*}
This is the only section of $(E^{[m]},N^m,\tau^{[m]})$ satisfying that $\pi_{E,j}\circ e^{[m]}=e\circ \pi_{N,j}$ for $j=\overline{1,m}$.
		
		Also of interest is the tensor field that transports vector fields from one copy of $N$ to another one within $N^m$. More specifically, let $T^{(1,1)} N$ denote the (1,1)-tensor bundle of the manifold $N$. For $r,s\in \overline{0, m-1}$, we define the $S_{rs}$ to be the sections of $T^{(1,1)} N^m$ over $N^m$ of the following form:
				\begin{equation}
		\label{eq:SrsDef}
		\textrm{If}\quad X^{[s]}:=\sum_{j=1}^nX_j(x^{(s)})\frac{\partial}{\partial x_j^{[s]}} \quad \textrm{then}\quad S_{rs} (X^{[s]}) := \sum_{j=1}^n X_j (x^{(s)}) \frac{\partial}{\partial x_j^{(r)}}, \quad \forall X \in \mathfrak{X} (N).
		\end{equation}
		In coordinates,  these tensor fields read
		\begin{equation}
		\label{eq:Srs}
		S_{rs} = \sum_{j=1}^n \d x_j^{(s)} \otimes \frac{\partial}{\partial x_j^{(r)}}, \quad r,s \in \overline{0, m-1}.
		\end{equation}
		These objects will play a key role in the computation of constants of motion and superposition rules.
		
		Finally, the {\it diagonal prolongation} of $f:N\rightarrow \mathbb{R}$ to $N^m$ is the function ${f}^{[m]}:N^m\rightarrow \mathbb{R}$ given by
		${f}^{[m]}(x^{(0)},\ldots,x^{(m-1)}):= f(x^{(0)})+\ldots+f(x^{(m-1)})$.

		\section{The geometrical description of quantum mechanics}
		We briefly present  the geometrical formulation of quantum mechanics
		which has been developed during the last forty years (see \cite{AS99,CCM07,CM08} for
		details).

		\subsection{The linear, complex and Hermitian structure}
		To investigate quantum mechanics in a differential geometric way and to identify its similarities with the geometric formalism of classical mechanics, the Hilbert space $\Hil$ must be understood 
		as a real 
		Banach manifold and its algebraic structures as real differential
		geometric objects in such manifolds. In particular, if $\Hil$ is of a complex dimension $n$, $\Hil$ should be identified with a real $2n$-dimensional differentiable manifold $\mathcal{M}_{2n}$. Each point $\psi\in \mathcal{M}_{2n}$ represents an element of $\Hil$. Any Hilbert basis in the Hilbert space $\mathcal{H}$ defines a real global chart on $\mathcal{H}$ which determines its differentiable structure. Let $\{e_j\}_{j\in\overline{1,n}}$ be an orthonormal basis of $\Hil$; the functions $q_j,p_j:\mathcal{M}_{2n}\rightarrow \mathbb{R}$ given by  
		\begin{equation}\label{CoorH}
		\langle e_j,\psi\rangle=q_j(\psi)+{\rm i}p_j(\psi),\qquad j=\overline{1,n},\quad \forall \psi \in \mathcal{M}_{2n},
		\end{equation}
		define a real global chart of $\mathcal{M}_{2n}$ and $T_\psi \mathcal{M}_{2n}=\langle \partial/\partial q_1,\partial/\partial p_1,\ldots,\partial/\partial q_n,\partial/\partial p_n\rangle$  at every $\psi\in \mathcal{M}_{2n}$.
		
		The complex structure on the $n$-dimensional Hilbert space $\mathcal{H}$, represented by the multiplication by the imaginary unit ${\rm i}$, can be  encoded in a (1,1)-tensor field $J$ on $\mathcal{M}_{2n}$ satisfying $J^{2}=-\mathbb{I}$ with $\mathbb{I}$ being the $(1,1)$-tensor field given by the identity $\mathbb{I}:T_\psi \mathcal{M}_{2n}\rightarrow T_\psi \mathcal{M}_{2n}$ at every $\psi \in \mathcal{M}_{2n}$. This leads to a distribution ${\rm Im}(T\mathcal{M}_{2n})$ on $\mathcal{M}_{2n}$, which is integrable. In the coordinate system (\ref{CoorH}), the  complex structure $J$ reads
		\begin{equation}
		\label{eq:J}
		J=\sum_{j=1}^n\left({\rm d}q_j\otimes \frac{\partial}{\partial
			p_j}-{\rm d}p_j\otimes \frac{\partial }{\partial q_j}\right).
		\end{equation}
		
		Another important element of $\Hil$ is its Hermitian product $\langle\cdot,\cdot\rangle :{\cal H}\times {\cal H}\rightarrow \mathbb{C}$. 
		As $\Hil\simeq \mathbb{R}^{2n}$ as $\mathbb{R}$-linear spaces,  there exists at each $\widehat \psi\in \mathcal{M}_{2n}$ an $\mathbb{R}$-linear isomorphism $\psi \in \mathcal{M}_{2n} \mapsto \,\psi_{\widehat \psi}\in T_{\widehat{\psi}} \mathcal{M}_{2n}$, where the tangent vector $\psi_{\widehat \psi}$ acts as a derivation 
		\begin{equation}
		\psi_{\widehat \psi} f := \left. \frac{\rm d}{{\rm d} t} \right|_{t=0}f\left(\widehat{\psi} + t \psi\right) , \qquad \forall f \in C^\infty (\mathcal{M}_{2n}).
		\end{equation}
		This identification and the Hermitian product on $\mathcal{H}$ allow us to define a pair of tensor fields $g, \omega$ on $\mathcal{M}_{2n}$ satisfying: 
		\begin{equation}\label{DefGW}
		g_\psi ({\psi_1}_\psi, {\psi_2}_\psi):=\mathfrak{Re}\langle \psi_1 , \psi_2 \rangle \qquad
		\omega_\psi ({\psi_1}_\psi,{\psi_2}_\psi):=\mathfrak{Im}\langle \psi_1 , \psi_2 \rangle, \qquad \forall \psi, \psi_1,\psi_2 \in \mathcal{M}_{2n}.
		\end{equation}
		These  tensor fields  encode in geometrical terms the Hermitian product existing in the Hilbert space. In coordinates, these tensor fields read
		\begin{equation}\label{SRStructure}
		g= \sum_{j=1}^n\left({\rm d}q_j \otimes {\rm d}q_j + {\rm d}p_j \otimes {\rm d}p_j\right), \quad
		\omega = \sum_{j=1}^n\left({\rm d}q_j \otimes  {\rm d}p_j - {\rm d}p_j \otimes  {\rm d}q_j\right)=\sum_{j=1}^n \d q_j \wedge \d p_j.
		\end{equation}
		The tensor field $g$ becomes a Euclidean metric on $\mathbb{R}^{2n}$ while $\omega$ becomes a symplectic structure on $\mathbb{R}^{2n}$ with Darboux coordinates $\{q_j,p_j\}_{j\in\overline{1,n}}$. The tensor fields $g$ and $\omega$ satisfy some relations with the complex structure $J$:
		\begin{equation*}
		g(JX,JY)=g(X,Y), \,\,
		\omega(JX,JY)=\omega(X,Y), \,\,
		\omega(X,Y)=g(JX,Y),\,\, \forall X,Y\in \mathfrak{X}(\mathcal{M}_{2n}).
		\end{equation*}
		Thus, the Hermitian product on the complex Hilbert space ${\cal H}$ leads to a K\"ahler structure on $\mathcal{M}_{2n}$, which is typical of quantum models and richer than the standard symplectic one typical appearing in classical mechanics.
		
		The metric tensor field $g$ induces a bundle isomorphism $\mathcal{G}: \psi_{\widehat{\psi}}\in T\mathcal{M}_{2n}\mapsto g( \psi_{\widehat{\psi}},\cdot)\in T^*\mathcal{M}_{2n}$. This can be used to transform $g$ and $\omega$ into two 2-contravariant tensor fields, i.e. $G(\alpha,\beta):=g(\mathcal{G}^{-1}\alpha,\mathcal{G}^{-1}\beta)$ and $\Lambda(\alpha,\beta):=\omega(\mathcal{G}^{-1}\alpha,\mathcal{G}^{-1}\beta)$. Their expressions in 
		local coordinates are
		\begin{equation}
		\label{eq:G}
		G=\sum_{j=1}^n \left( \frac{\partial}{\partial q_j}\otimes \frac{\partial}{\partial
			q_j}+\frac{\partial}{\partial p_j}\otimes \frac{\partial}{\partial
			p_j} \right),\qquad
		\Lambda=
			\sum_{j=1}^n\partd{q_j} \wedge \partd{p_j}.
		\end{equation}
		These tensor fields define a Poisson bracket and a commutative bracket on $C^\infty(\mathcal{M}_{2n})$, respectively:
		\begin{equation}\label{Dpsi}
		\{ f, g\} := \Lambda({\rm d}f, {\rm d}g), \qquad
		\{f, g\}_{+} := G({\rm d}f, {\rm d}g), \qquad
		\forall f,g \in C^\infty(\mathcal{M}_{2n}).
		\end{equation}
		
		A third element in the description of the Hilbert space structure of $\Hil$ is its $\mathbb{R}$-linear structure. Geometrically, it is induced by the so-called {\it dilation vector field} defined by $\Delta:\psi\in \mathcal{M}_{2n}\mapsto  \psi_\psi \in T \mathcal{M}_{2n}$. Meanwhile, the {\it phase-change vector field} takes the form $\Gamma:\psi\in \mathcal{M}_{2n} \mapsto J\psi_{\psi}\in T_\psi \mathcal{M}_{2n}$.		In local coordinates
		\begin{equation}
		\label{eq:delta}
		\Delta =\sum_{j=1}^n\left( q_j\frac{\partial}{\partial q_j}+p_j\frac{\partial}{\partial p_j}\right), \qquad \Gamma =\sum_{j=1}^n\left( q_j \frac{\partial}{\partial p_j} - p_j \frac{\partial}{\partial q_j}\right).
		\end{equation}
		Both vector fields satisfy the relation $\Gamma=J(\Delta)$.
		
		Finally, let us define two $n$-forms, $\Omega_R$ and $\Omega_I$, on $\mathcal{M}_{2n}$ satisfying
		$$
		\Omega_R |_\psi ({\psi_1}_\psi,\ldots,{\psi_n}_\psi):={\mathfrak Re} (\det(\psi_1,\ldots,\psi_n)), \quad
		\Omega_I |_\psi ({\psi_1}_\psi,\ldots,{\psi_n}_\psi):={\mathfrak Im} (\det(\psi_1,\ldots,\psi_n)),
		$$
		for all $ \psi, \psi_1,\ldots,\psi_n\in \mathcal{M}_{2n}.$
		They satisfy the relation
		$$
		\Omega_R(JX_1,\ldots,X_n)=-\Omega_I(X_1,\ldots,X_n), \quad \forall X_1, \ldots, X_n \in \mathfrak{X}(\mathcal{M}_{2n}).
		$$
		It is simple to prove that they are non-degenerate and closed. 

		\subsection{Observables: Hamiltonian dynamics and Killing vector fields}\label{ODG}

		The real vector space ${\rm Herm}(\mathcal{H})$ of physical observables on $\mathcal{H}$, i.e. Hermitian operators on $\Hil$, can also be given a tensor description. Every observable $A\in {\rm Herm}(\Hil)$ gives rise to a function on
		$\mathcal{M}_{2n}$ of the form
		\begin{equation}
		\label{eq:14}
		f_A(\psi):= \frac 12 \langle \psi, A\psi \rangle,\qquad \psi\in \mathcal{M}_{2n}.
		\end{equation}

		It is worth noting that
		\begin{equation*}
		\label{eq:19}
		\{ f_{A}, f_{B}\} = \Lambda(\d f_{A}, \d f_{B}) = f_{[\![ A,B ]\!] },\qquad \qquad
		\{ f_{A}, f_{B}\}_{+} = G(\d f_{A}, \d f_{B}) = f_{[A,B]_{+}},
		\end{equation*}
		for 
		\begin{equation}
		\label{eq:18}
		[\![ A,B ]\!]:=-{\rm i} (AB-BA),	 \qquad [A,B]_{+}:= AB+BA.
		\end{equation}
		Observe that the skew-symmetric operation has an extra factor with respect to the commutator of operators. This factor is needed to obtain an inner composition law in the space of Hermitian operators.
		Given a linear operator $A\in \mathrm{Herm}({\cal H})$, we can define the Hamiltonian vector field:
		\begin{equation}
		\label{eq:10}
		X_A:=-{\Lambda}({\rm d}f_{A}, \cdot)=\{ \cdot,f_{A}\}.
		\end{equation}
		One remarkable property of these Hamiltonian vector fields is
		\begin{equation*}
		\label{eq:20}
		[X_A, X_B]=-X_{[\![ A,B ]\!]}. 
		\end{equation*}
		The vector field $\Gamma$ defined
		in \eqref{eq:delta} is the Hamiltonian vector field associated, up to a sign, with the identity operator on $\Hil$, i.e. $\Gamma= -X_I$.
		
		
%
		The integral curves of the Hamiltonian vector 
		field, $X_H$, associated with the quadratic form $f_{H}(\psi):= \frac 12 \langle \psi , H \psi
		\rangle$ correspond to the solutions of the Schr\"odinger equation 
		$$
		{\rm i}\frac{\d \psi}{\d t}=H\psi,
		$$
		where we assumed, as hereafter, $\hbar=1$. The evolution operator $t\mapsto U_t$ of this equation is such that each $U_t:\mathcal{H}\rightarrow\mathcal{H}$ is an isometry of the Hermitian product on $\mathcal{H}$. Hence, each $U_t$ leaves invariant its real and imaginary parts. Since each $U_t$ is $\mathbb{C}$-linear, it also leaves invariant $\omega$, $J$, and $g$. Therefore, $X_H$ is also a Killing vector field relative to $g$ giving rise to a K\"ahler vector field.
%

		\subsection{Projective Hilbert spaces as K\"ahler manifolds}
		
		From a physical point of view, the
		probabilistic interpretation requires the set of states of a quantum
		system to be a complex projective space. Our aim 
		in this section is to introduce the geometrical structures arising in this case. 
		To simplify the notation, $B^\times$ will stand for the restriction of a structure $B$ on $\mathcal{M}_{2n}$ to $\mathcal{M}_{2n}^\times$, e.g. $G^\times$ is the tensor field on $\mathcal{M}_{2n}^\times$ obtained by restricting the tensor $G$ on $\mathcal{M}_{2n}$ given in (\ref{eq:G}) to $\mathcal{M}_{2n}^\times$. 

		The equivalence relation on {$\mathcal{H}_0:=\mathcal{H}\backslash\{0\}$ defining its} projective space $\mathbb{CP}^{n-1}$, namely
		\begin{equation}
		\label{eq:11}
		\psi_{1},\psi_{2} \in \Hil_0 := \Hil\backslash\{0\}, \qquad \psi_{1}\sim
		\psi_{2}\Leftrightarrow \psi_{2}=\lambda \psi_{1}, \qquad \lambda\in \C\backslash\{0\},
		\end{equation}
		can be encoded at the level of our geometrical description by means of
		the vector fields $\Delta^\times$ and $\Gamma^\times$. Indeed, it follows from (\ref{eq:delta}) that both vector fields commute and define a regular integrable distribution on $\mathcal{M}^\times_{2n}$. Its space of leaves can be given a differentiable manifold structure becoming a differentiable manifold ${\cal P}_n$ and inducing a differentiable projection
		$\pi_{\mathcal{MP}}:\mathcal{M}^\times_{2n}\to {\cal P}_n$.  Each leaf of the foliation induced by $\Delta^\times$ and $\Gamma^\times$ contains the set of
		equivalent points in $\mathcal{H}_0$ relative to the equivalence relation
		(\ref{eq:11}). Therefore, it is natural to consider
		${\cal P}_n$ as the geometrical representation of the complex projective
		space $\mathbb{CP}^{n-1}$.
		
		In what regards the tensor structures, it is well known \cite{Helg78} that the complex projective space is a Hermitian
		symmetric space \cite{Wo64} and therefore admits  a canonical K\"ahler structure enconded in the Fubiny-Study
		metric. 
		Hence, there will exist a Riemannian tensor, a symplectic tensor and a complex structure on $\mathcal{P}_n$.  Our goal now
		is to relate these tensorial objects to those of
		$\mathcal{M}^\times_{2n}$.
		We will
		find suitable tensor fields on $\mathcal{M}^\times_{2n}$ projecting under $\pi_{\mathcal{MP}}:\mathcal{M}^\times_{2n}\to {\cal P}_n$ onto the
		canonical K\"ahler structure on ${\cal P}_n$ induced by that of $\mathbb{CP}^{n-1}$. 
		
		The tensor fields $G^\times$ and $\Lambda^\times$ on $\mathcal{M}_{2n}^\times$ induced by \eqref{eq:G} cannot be projected onto ${\cal P}_n$ since, although they are invariant under $\Gamma^\times$,  they are
		not invariant under $\Delta^\times$. Since
		${\cal L}_{\Delta^\times}G^\times=-2G^\times$ and ${\cal L}_{\Delta^\times}\Lambda^\times=-2\Lambda^\times$,
		we can define two new tensor fields by multiplication:
		$$
		\widetilde G:=2f_I(\psi)G,\qquad \qquad \widetilde \Lambda:=2f_I(\psi)\Lambda.
		$$
		The tensor fields $\widetilde G^\times$ and $\widetilde \Lambda^\times$ on $\mathcal{M}^\times_{2n}$ are homogeneous of degree 0 and therefore invariant
		under dilations. They are also invariant under $\Gamma^\times$ and hence {projectable}  onto $\mathcal{P}_n$. Nevertheless, we would like to define projectable tensor fields $G_\mathcal{P}$ and $\Lambda_\mathcal{P}$ on $\mathcal{M}^\times_{2n}$ satisfying that 
		\begin{align*}
		G_\mathcal{P}(\d \pi^*f_1,\d \pi^*f_2):=\{\pi_{\mathcal{MP}}^* (f_1),\pi_{\mathcal{MP}}^* (f_2) \}_+,\qquad
		\Lambda_\mathcal{P} (\d \pi^*f_1,\d \pi^*f_2):=\{\pi_{\mathcal{MP}}^* (f_1),\pi_{\mathcal{MP}}^* (f_2)\}, 
		\end{align*}
		for every $f_1, f_2 \in C^\infty(\P_n).$	To ensure this, we define
		\begin{equation}
		\label{eq:13}
		G_{{\cal P}}:=2 f^\times_I(\psi) G^\times-\left (\Delta^\times\otimes
		\Delta^\times + \Gamma^\times\otimes \Gamma^\times \right ), \quad 
		\Lambda_{{\cal P}}:= 2f^\times_I(\psi)\Lambda^\times-\left (\Delta^\times\otimes
		\Gamma^\times - \Gamma^\times\otimes \Delta^\times \right ).
		\end{equation}
		
		Observables must also be represented on $\mathcal{P}_n$ by tensorial objects which are projections of tensor objects on $\mathcal{M}_{2n}^\times$ that are invariant by the vector fields $\Delta^\times$ and $\Gamma^\times$, and this is
		clearly not the case for the quadratic functions $f^\times_{A}$ defined in (\ref{eq:14}). 
		Instead, we consider the set of  {\it expectation value functions}, i.e. the
		functions related to observables $A$ of the form
		\begin{equation*}
		e_A(\psi):=\frac 12\frac{\langle \psi, A\psi \rangle}{\langle \psi,\psi\rangle},\qquad \forall \psi\in \mathcal{M}_{2n}^\times.
		\label{eq:5}
		\end{equation*}
		These functions are first-integrals on $\mathcal{M}_{2n}^\times$ of  both vector fields, $\Delta^\times$ and $\Gamma^\times$, and as they are projectable, they correspond  to 
		pullbacks of functions on ${\cal P}_n$. Furthermore
		they represent, up to a proportional constant, the physical magnitude known as \textit{expectation
			value} of the observable $A$.
		
		Finally, we can combine the expectation value functions and the
		tensor $\Lambda_{{\cal P}}$ to define Hamiltonian vector fields on ${\cal P}_n$. Indeed, given a Hermitian
		operator $A$, we can define the vector field on $\mathcal{M}^\times_{2n}$:
		\begin{equation*}
		\label{eq:15}
		{X}_{A}:=- \Lambda_{{\cal P}}({\rm d}e_{A}, \cdot).
		\end{equation*}
		The projections of these vector fields under $\pi_{\mathcal{MP}*}$ give rise to Hamiltonian
		vector fields associated with the canonical K\"ahler
		structure on ${\cal P}_n$. 
		
		Additionally, there exists a natural action of the unitary group $U(n)$ on ${\cal P}_n$ of the form
		\begin{equation}
		\label{eq:17}
		\varphi_{\mathcal{P}_n}:U({n})\times {\cal P}_n \to {\cal P}_n,  \qquad (U,
		[\psi]_{\mathcal{P}_n})\mapsto [U\psi]_{\mathcal{P}_n}, 
		\end{equation}
		where $[\psi]_{\mathcal{P}} := \pi_{\mathcal{MP}}(\psi)$ denotes the equivalence class in $\mathcal{P}_n$ of the element $\psi\in \mathcal{M}_{2n}^\times$.

		\section{Quantum Lie systems and K\"ahler--Lie systems}
		
		In this section we apply the theory of Lie and Lie--Hamilton systems to a $t$-dependent
		Hamiltonian operator $H(t)$ that can be written as a linear combination, with some $t$-dependent
		real coefficients $b_1(t),\ldots,b_r(t)$, of some Hermitian operators,
		\begin{equation}\label{LieHamiltonian}
		H(t)=\sum_{k=1}^rb_k(t)H_k\,,
		\end{equation}
		where the
		$H_k$ form a basis of a real finite-dimensional Lie algebra $V$ relative to the Lie bracket of observables, i.e.
		$[\![ H_j,H_k ]\!]=\sum_{l=1}^rc_{jkl} H_l$, with  $c_{jkl}\in\mathbb{R}$ and $j,k,l=\overline{1,r}$. The $t$-dependent operator $H(t)$, a so-called {\it quantum Lie system} \cite{CLR09}, becomes a curve in a Lie algebra of operators: {\it quantum VG--Lie algebra} of $H(t)$. In particular, we prove that a very general class of these systems leads to define Lie systems admitting a VG--Lie algebra of K\"ahler vector fields with respect to a K\"ahler structure. In turn, this suggests us to define a new type of Lie systems: the K\"ahler--Lie systems.

		
		In particular, a quantum Lie system $H(t)$ determines a $t$-dependent Schr\"{o}dinger equation
		\begin{equation}\label{SchLie}
		\frac{{\rm d}\psi}{{\rm d}t}=-{\rm i}H(t)\psi=-{\rm i}\sum_{k=1}^rb_k(t)H_k\psi.
		\end{equation}
		The isomorphism $(\psi,\phi)\in \mathcal{M}_{2n} \oplus \mathcal{M}_{2n} \mapsto \phi_\psi \in {\rm T}_\psi \mathcal{M}_{2n}$ allows us to identify the operators $-{\rm i}H_k$ with the vector fields
		$X_k:\psi\in \mathcal{M}_{2n} \mapsto (\psi,-{\rm i}H_k \psi) \in
		\mathcal{M}_{2n} \oplus \mathcal{M}_{2n} \simeq T\mathcal{M}_{2n}$.
		In turn, the $t$-dependent 
		Schr\"odinger equation (\ref{SchLie}) becomes the associated system
		of the $t$-dependent vector field $X =\sum_{k=1}^r b_k(t)X_k$ on $\mathcal{M}_{2n}$. 
		
		It was stated in Section \ref{ODG} that  $\mathcal{M}_{2n}$
		admits a natural symplectic structure turning the vector fields $X_k$
		into Hamiltonian admitting  real Hamiltonian functions
		$h_k(\psi)=\frac 12 \langle \psi,H_k\psi\rangle$. From \eqref{eq:20}, the commutators of these Hamiltonian vector fields are
		\begin{equation}
		[X_j,X_k]=-X_{[\![ H_j,H_k ]\!]}=-\sum_{l=1}^rc_{jkl}X_l,\qquad j,k=\overline{1,r}.
		\end{equation}
		Hence, $t$-dependent  Schr\"odinger
		equations (\ref{SchLie}) are Lie--Hamilton systems.  Summing up, we
		have this first theorem. 
		
		\begin{theorem} Every $t$-dependent Schr\"odinger equation on $\mathcal{M}_{2n}$ determined by a quantum Lie system
			$H(t)$ is a Lie--Hamilton system. 
		\end{theorem}
		
		\begin{note}
			From now on we will only consider Schr\"odinger equations related to a quantum--Lie system.
		\end{note}
		We exemplify the above result by studying  a two-level quantum
		system. Its possible states are described by elements $\psi\in
		\mathbb{C}^2$. 
		Unitary evolution is described by the canonical action of the unitary Lie group
		$U(2)$ on $\mathbb{C}^2$. In consequence, the
		evolution of every particular solution $\psi(t)$ in $\mathbb{C}^2$ of
		the corresponding $t$-dependent Schr\"odinger equation is determined
		by a curve $t\mapsto U_t$ within $U(2)$ which in turn gives rise to a curve
		$-{\rm i}H(t):=\dot U_tU_t^{-1}$ in the Lie algebra $\mathfrak{u}(2)$ of
		$U(2)$. More specifically, 
		\begin{equation}
		\label{Schr}
		\frac{{\rm d}\psi}{{\rm d}t} = -{\rm i} H(t) \psi, \qquad -{\rm i}H(t) \in \mathfrak{u}(2),\qquad \forall t\in\mathbb{R}.
		\end{equation}
		As $\mathfrak{u}(2)$ is the space of skew-Hermitian operators on
		$\mathbb{C}^2$, each $H(t)$ is Hermitian.
		
		Consider a basis for ${\rm Herm} (2)$ given by the $2\times2$ identity matrix $I_0$ and the traceless matrices $\{S_j := \frac 12 \sigma_j\}_{j=1,2,3}$,
		where $\sigma_1,\sigma_2,\sigma_3$ are the Pauli matrices, namely
		\begin{equation*}I_0 =\begin{pmatrix}
		1 & 0 \\
		0 & 1
		\end{pmatrix},\quad
		\sigma_1 =  \begin{pmatrix}
		0 & 1 \\
		1 & 0
		\end{pmatrix}, \quad
		\sigma_2 = \begin{pmatrix}
		0 & -i \\
		i & 0
		\end{pmatrix}, \quad
		\sigma_3 = \begin{pmatrix}
		1 & 0 \\
		0 & -1
		\end{pmatrix}.
		\end{equation*}
		
		If we consider the commutator defined in \eqref{eq:18}, then
		we find that $[\![ I_0,\cdot ]\!]=0$ and 
		\begin{equation}
		\label{LiePSigma}
		[\![ S_j, S_k ]\!] =\sum_{l=1}^3 \epsilon_{jkl} S_l,\qquad j,k=1,2,3.
		\end{equation}
		
		Every Hamiltonian in ${\rm Herm}(2)$ can be brought into the form:
		\begin{equation}\label{Hamiltonian}
		H = B_0I_0+\sum_{j=1}^3 B_j S_j = B_0I_0+{\bf B} \cdot {\bf S} , \,\,\, {\bf S}  = (\sigma_1, \sigma_2, \sigma_3)/2, \,\,\, {\bf B}:= (B_1, B_2, B_3) \in \mathbb{R}^3,\,\,\, B_0\in\mathbb{R}.
		\end{equation}
		
		In a physical system, {${\bf B} $} is identified with the magnetic field
		applied to a 1/2-spin particle. To obtain a $t$-dependent Hamiltonian,
		the magnetic field must be $t$-dependent: 
		\begin{equation}
		\label{eq:Ht}
		H (t) :=B_0(t)I_0+ {\bf B}  (t) \cdot {\bf S} .
		\end{equation}
		The $t$-dependent Hamiltonian $H(t)$ is therefore a quantum Lie system. It determines a $t$-dependent Schr\"odinger equation of the form \eqref{Schr} in $\mathbb{C}^2$  \cite{CGM01}. 
		
		Consider now the geometric formalism presented in the previous section. The Hilbert space $\mathbb{C}^2$ is replaced by a real manifold $\mathcal{M}_4$ with coordinates $(q_1, p_1, q_2, p_2)$. The coordinate expression of the $t$-dependent Schr\"odinger equation \eqref{Schr} with the quantum Lie system defined by \eqref{eq:Ht} is
		\begin{equation}\label{C2}
		\frac{\rm d}{{\rm d}t}\!\left[\begin{array}{c}q_1\\p_1\\q_2\\p_2\end{array}\right]\!=\!
		\frac 12\!\left[\begin{array}{cccc}0&2B_0(t)+B_3(t)&-B_2(t)&B_1(t)\\
		-2B_0(t)-B_3(t)&0&-B_{1}(t)&-B_2(t)\\
		B_{2}(t)&B_{1}(t)&0&2B_0(t)-B_3(t)\\
		-B_{1}(t)&B_2(t)&B_3(t)-2B_0(t)&0
		\end{array}\right]\!\left[\begin{array}{c}q_1\\p_1\\q_2\\p_2\end{array}\right].
		\end{equation}
		This is the associated system of the $t$-dependent vector field $X =\sum_{\alpha=0}^3 B_\alpha(t)X_\alpha$, with 
		{\small\begin{equation}
		\label{2levelXj} 
		\begin{gathered}
		X_0 = -\Gamma={p_1}\frac{\partial}{\partial
			{q_1}}-{q_1}\frac{\partial}{\partial
			{p_1}}+{p_2}\frac{\partial}{\partial
			{q_2}}-{q_2}\frac{\partial}{\partial {p_2}},\quad 
		X_1=\frac 12\left({p_2}\frac{\partial}{\partial
			{q_1}}-{q_2}\frac{\partial}{\partial
			{p_1}}+{p_1}\frac{\partial}{\partial
			{q_2}}-{q_1}\frac{\partial}{\partial {p_2}}\right), \\ 
		X_2 = \frac 12 \left( -{q_2}\frac{\partial}{\partial {q_1}}
		-{p_2}\frac{\partial}{\partial {p_1}} +{q_1}\frac{\partial}{\partial
			{q_2}} + {p_1}\frac{\partial}{\partial {p_2}}\right),\quad  
		X_3=\frac 12\left({p_1}\frac{\partial}{\partial
			{q_1}}-{q_1}\frac{\partial}{\partial
			{p_1}}-{p_2}\frac{\partial}{\partial
			{q_2}}+{q_2}\frac{\partial}{\partial {p_2}}\right), 
		\end{gathered}
		\end{equation}}
		spanning a Lie algebra of vector fields isomorphic to $\mathfrak{u}(2)$:
		$$
		[X_0,\cdot]=0,\qquad [X_1,X_2]=-X_3,\qquad [X_2,X_3]=-X_1,\qquad [X_3,X_1]=-X_2.
		$$
		
		Recall that $\mathcal{M}_4$ admits a K\"ahler structure with symplectic and Riemannian tensor fields given by 
		$$
		\omega=\sum_{j=1}^2{\rm d}{q_j}\wedge {\rm d}{p_j},\qquad g=\sum_{j=1}^2({\rm d}{q_j}\otimes {\rm d}{q_j}+{\rm d}{p_j}\otimes {\rm d}{p_j}).
		$$
		The vector fields $X_0,X_1,X_2,X_3$ are Hamiltonian with respect to $\omega$. Their Hamiltonian functions are
		\begin{equation}
		\label{eq:hj}
		\begin{gathered}
		h_0 (\psi) = \frac 12 \langle \psi,\psi\rangle = \frac 12 (q_1^2+p_1^2+q_2^2+p_2^2),\quad
		h_1 (\psi) = \frac 12 \langle \psi,S_1\psi\rangle = \frac 12 (q_1 q_2 + p_1 p_2), \\
		h_2 (\psi) = \frac 12 \langle \psi,S_2\psi\rangle = \frac 12 (q_1 p_2 - p_1 q_2), \quad
		h_3 (\psi) = \frac 12 \langle \psi,S_3\psi\rangle = \frac 14 (q_1^2 + p_1^2 -q_2^2 -p_2^2),
		\end{gathered}
		\end{equation}
		with $\iota_{X_\alpha}\omega= \d h_\alpha$ for $\alpha = 0,1,2,3$.
	 
	 	The Hamiltonian functions span a Lie algebra isomorphic to $\mathfrak{u}(2)$:
		$$
		\{h_0,\cdot\}=0,\qquad \{h_1,h_2\}=h_3,\qquad \{h_2,h_3\}=h_1,\qquad \{h_3,h_1\}=h_2.
		$$
		It will be useful to note that $h_1,h_2,h_3$
		are functionally independent, but $h_0^2 = 4(h_1^2+h_2^2+h_3^2)$.
		
		The $t$-dependent Schr\"odinger equation (\ref{Schr}) enjoys an additional
		property: $X_0, X_1,X_2$ and $X_3$ are Killing vector fields with respect to
		$g$, namely $\mathcal{L}_{X_\alpha}g=0$ for $\alpha=0,1,2,3$. Using
		this, we can easily prove in an intrinsic geometric way that 
		$$
		I_1=g(X_0,X_0),\qquad I_2= g(X_1,X_1)+g(X_2,X_2)+g(X_3,X_3),\qquad I_3=h_1^2+h_2^2+h_3^2,\qquad I_4=h_0
		$$
		are constants of the motion for $X$. This example is relevant because it illustrates how to define the above constants of the motion geometrically in terms of $g$ and the Hamiltonian functions due to $\omega$. 
		
		Note also that our {real} system comes from a linear complex differential equation. This gives rise to a symmetry $(q_1,p_1,q_2,p_2) \in \mathcal{M}_4 \mapsto (-p_1,q_1,-p_2,q_2) \in \mathcal{M}_4$ of system (\ref{C2}), which is the counterpart of the multiplication by the imaginary unit in $\mathbb{C}^2$. Therefore, the Lie system  preserves the complex structure $J$ in $\mathcal{M}_4$.
		
		Finally, $\mathcal{M}_{4}$ admits the following symplectic forms 
		$$
		\Omega_R:={\rm d}q_1\wedge {\rm d}q_2 - {\rm d}p_1 \wedge {\rm d}p_2,\qquad \Omega_I:={\rm d}q_1 \wedge {\rm d}p_2 + {\rm d}p_1 \wedge {\rm d}q_2
		$$
		turning $X_1,X_2,X_3$ into Hamiltonian vector fields. However, these symplectic forms are not invariant under $X_0$:
		\begin{equation}
		\mathcal{L}_{X_0} \Omega_R = \Omega_I, \quad
		\mathcal{L}_{X_0} \Omega_I = - \Omega_R.
		\end{equation}
		Recall that $X_0 = -\Gamma$. This result proves that $\Omega_R$ and $\Omega_I$ are invariant under the canonical $SU(2)$-action on $\mathcal{M}_4$, but not under the $U(2)$-action.
		
		The results here presented can be used to obtain a superposition rule for the initial system. This topic will be studied in following sections. Let us generalise the above example. 
		
		\begin{theorem}\label{KahH} Every Schr\"odinger equation on $\mathcal{M}_{2n}$
	admits a VG--Lie algebra
			$V_{\mathcal{M}_{2n}}\simeq\mathfrak{u}(n)$ of K\"ahler vector fields relative to the
			K\"ahler structure $(g,\omega, J)$ on $\mathcal{M}_{2n}$.   
		\end{theorem}
		\begin{proof} Let $\varphi_{\mathcal{M}_{2n}}:U(n)\times \mathcal{M}_{2n}\rightarrow \mathcal{M}_{2n}$ be
			the natural action of the unitary group by unitary matrices on $\mathbb{C}^n$ understood in the natural way as a
			real manifold $\mathcal{M}_{2n}$.
			
			Every element $h\in U(n)$ induces a diffeomorphism on
			$\mathcal{M}_{2n}$  leaving invariant the Hermitian product on $\mathcal{H}$. Hence, it leaves invariant its real and imaginary parts.
			In view of expressions (\ref{DefGW}), it also follows that $h_t^*\omega=\omega$ and $h_t^*g=g$ for every $t\in\mathbb{R}$ and every curve $h_t$ in $U(n)$. As a
			consequence, every fundamental vector field $Y$ of the Lie group
			action $\varphi_{\mathcal{M}_{2n}}$ satisfies that $\mathcal{L}_Yg=0$ and
			$\mathcal{L}_Y\omega=0$. Hence, the fundamental vector fields
			of $\varphi_{\mathcal{M}_{2n}}$ are K\"ahler vector fields relative to $(g,\omega,J)$. 
			
			The evolution of the Schr\"odinger equation   
			\begin{equation}\label{SchUni}
			\frac{{\rm d}\psi}{{\rm d}t}=-{\rm i}H(t)\psi,\qquad -{\rm i} H(t)\in \mathfrak{u}(n),\qquad \psi\in \mathcal{M}_{2n}, 
			\end{equation}
			takes the form $\psi(t)=\varphi_{\mathcal{M}_{2n}}(h_t,\psi(0))$, with $h_0={\rm Id}_{\mathcal{M}_{2n}}$, and a certain curve $h: t\in \mathbb{R}\mapsto h_t\in U(n)$. 
			Therefore, (\ref{SchUni}) is determined by a 
			$t$-dependent vector field $X$ taking values in the Lie algebra of fundamental vector fields of $\varphi_{\mathcal{M}_{2n}}$, namely
			$			X =\sum_{\alpha=1}^r b_\alpha(t)X_\alpha,
			$
			where $X_1,\ldots,X_r$ is a basis of the Lie algebra of fundamental vector fields of $\varphi_{\mathcal{M}_{2n}}$.  Then $\langle X_1,\ldots,X_r\rangle$ is a Lie algebra isomorphic to $\mathfrak{u}(n)$ and becomes a VG--Lie algebra of K\"ahler vector fields for $X_{\mathcal{M}_{2n}}$.
		\end{proof}
		
		We hereafter write $V_{\mathcal{M}_{2n}}$ and $V_{\mathcal{M}_{2n}^\times}$ for VG--Lie algebras of K\"ahler vector fields relative to the standard K\"ahler structure on $\mathcal{M}_{2n}$ and its natural restriction to $\mathcal{M}_{2n}^\times$, respectively.
		The above examples motivate to introduce the following definition.
		
		\begin{definition} We call {\it K\"ahler--Lie system} a Lie system admitting a VG--Lie algebra of K\"ahler vector fields with respect to a K\"ahler structure.
		\end{definition}
		
		It is therefore simple to prove the proposition below.
		\begin{proposition} The space $I_X$ of $t$-independent constants of motion for a K\"ahler--Lie system $X$ is a Poisson algebra with respect to the Poisson bracket of the  K\"ahler structure and the commutative algebra relative to the bracket induced by its Riemannian structure.
		\end{proposition}
		In following sections, it will be shown that K\"ahler structures allow us to devise techniques to obtain constants of the motion and superposition rules for K\"ahler--Lie systems, e.g. if vector fields $Y_1$, $Y_2$ commute with all the elements of the VG--Lie algebra of K\"ahler vector fields for a K\"ahler--Lie system $X$, then $g(Y_1,Y_1)$, $g(Y_2, Y_2)$, $g(Y_1, Y_2)$ and $\omega (Y_1, Y_2)$ are constants of motion for $X$.

		\section{Lie systems and Schr\"odinger equations on $\mathcal{R}_n$ and $\mathcal{S}_{n}$}
		
		The unit sphere $\mathcal{S}_n$ in $\mathbb{C}^n$ admits a natural structure as a real $(2n-1)$-dimensional manifold. Since we study $t$-dependent Schr\"odinger equations with a unitary evolution and their evolution leave $\mathcal{S}_n$ invariant, it seems natural at first to restrict them to the unity sphere $\mathcal{S}_n$. Nevertheless, as shown next, their restriction is generally no longer neither a K\"ahler--Lie  system nor a Lie--Hamilton one. That is why we now introduce an alternative Schr\"odinger equation which possesses more useful properties to describe its superposition rules and the superposition rules for other related Schr\"odinger equations.
		
		
		\begin{proposition}\label{ResSch} A Schr\"odinger equation on $\mathcal{M}_{2n}$ can be restricted to the unity sphere $\mathcal{S}_{n}$ giving rise to a Lie system $X_{\mathcal{S}_{n}}$ possessing a VG--Lie algebra $V_{\mathcal{S}_{n}}$ of Hamiltonian vector fields with respect to the presymplectic form $\iota_\mathcal{S}^*\omega$ with $\iota_\mathcal{S}:\mathcal{S}_n\hookrightarrow \mathcal{M}_{2n}$. If $V^{X_{\mathcal{S}_n}}=V_{\mathcal{S}_{n}}$, then $X_{\mathcal{S}_n}$  is not a Lie--Hamilton system.
		\end{proposition}
		\begin{proof} 
			The VG--Lie algebra $V_{\mathcal{M}_{2n}}$ of (\ref{SchUni}) is the Lie algebra of fundamental vector fields of the unitary action of $\varphi_{\mathcal{M}_{2n}}:U(n)\times \mathcal{M}_{2n}\rightarrow \mathcal{M}_{2n}$. Therefore, $\langle \varphi_{\mathcal{M}_{2n}}(g,\psi),\varphi_{\mathcal{M}_{2n}}(g,\psi)\rangle=\langle \psi,\psi\rangle$ for every $\psi\in \mathcal{M}_{2n}$ and $g\in U(n)$. Hence, $f(\psi):=\langle \psi,\psi\rangle$ is invariant under $\varphi_{\mathcal{M}_{2n}}$  and, in consequence, a first-integral of its fundamental vector fields, namely $V_{\mathcal{M}_{2n}}$. The restrictions of the elements of $V_{\mathcal{M}_{2n}}$ to $\mathcal{S}_{n}$ become tangent to $\mathcal{S}_{n}$ and they therefore span a finite-dimensional Lie algebra of vector fields $V_{\mathcal{S}_n}$ on $\mathcal{S}_n$. In the light of Theorem \ref{KahH}, the system (\ref{SchUni}) is related to a $t$-dependent vector field $X_{\mathcal{M}_{2n}}$ taking values in $V_{{\mathcal{M}_{2n}}}$. Therefore, system (\ref{SchUni}) can be restricted to a system $X_{\mathcal{S}_n}$ on $\mathcal{S}_{n}$ admitting a VG--Lie algebra $V_{\mathcal{S}_n}$.
			
			The embedding $\iota_{\mathcal{S}} :\mathcal{S}_{n}\hookrightarrow \mathcal{M}_{2n}$ gives rise to a presymplectic structure $\iota_{\mathcal{S}}^*\omega$ on $\mathcal{S}_{n}$, where $\omega$ is the natural symplectic structure (\ref{SRStructure}) on $\mathcal{M}_{2n}$. Since the elements of $V_{\mathcal{M}_{2n}}$ are Hamiltonian vector fields on $\mathcal{M}_{2n}$ with Hamiltonian functions $h_H(\psi) = \frac 12 \langle \psi,H\psi\rangle$ with $-{\rm i}H\in \mathfrak{u}(n)$,  their restrictions to $\mathcal{S}_n$ are tangent to $\mathcal{S}_{n}$ and Hamiltonian relative to the presymplectic form $\iota_{\mathcal{S}}^*\omega$ with Hamiltonian functions $\iota_{\mathcal{S}}^* h_H$.  Therefore,
they span a VG--Lie algebra $V_{\mathcal{S}_n}$ on $\mathcal{S}_n$ of Hamiltonian vector fields relative to the presymplectic structure $\omega_S$. 
			
			As $\mathcal{S}_{n}$ is an orbit of $\varphi_{\mathcal{M}_{2n}}$, then $T\mathcal{S}_{n}=\mathcal{D}^{V_{\mathcal{M}_{2n}}}|_{\mathcal{S}_{n}}$, which is an odd $(2n-1)$-dimensional distribution on $\mathcal{S}_n$. From assumption $V^{X_{\mathcal{S}_n}}=V_{\mathcal{S}_n}$ and, hence, $\mathcal{D}^{X_{\mathcal{S}_n}}=\mathcal{D}^{V_{\mathcal{S}_n}}=\mathcal{D}^{V_{\mathcal{M}_{2n}}}|_{\mathcal{S}_{n}}=T\mathcal{S}_n$. The so-called no-go Theorem for Lie--Hamilton systems (see \cite[Proposition 5.1]{CGLS13}) states that previous conditions are enough to ensure that $X_{\mathcal{S}_n}$ is not a Lie--Hamilton system.
\end{proof}
		A Dirac structure is a generalisation of presymplectic and Poisson manifolds. In fact, presymplectic and Poisson manifolds can be naturally attached to Dirac structures whose Hamiltonian vector fields are the Hamilton vector fields of the structures originating them (see  \cite{CGLS13} for details). This fact enables us to prove the following.
		\begin{corollary} The Schr\"odinger equation (\ref{SchUni}) on ${\mathcal{S}_{n}}$ is a Dirac--Lie system with respect to the Dirac structure induced by $\iota_\mathcal{S}^*\omega$. 
		\end{corollary}
		
		Let us now prove that the projection of the restriction of (\ref{SchUni}) to $\mathcal{M}_{2n}^\times$ onto $\mathcal{R}_n$ exists and it is a Lie--Hamilton system that can be endowed with a natural coordinate system coming from this fact.
%
		\begin{lemma}\label{CoorHU}
			The manifold $\mathcal{R}_n$, for $n>1$, admits a local coordinate system on a neighborhood of each point  given by $2n-1$ functions $f_\alpha (\psi) = \frac 12 \langle \psi,H_\alpha\psi\rangle$, for $\alpha=\overline{1,2n-1}$ for certain operators $H_\alpha\in \mathfrak{su}(n)$. 
		\end{lemma}
		\begin{proof} For $n>1$ any two elements of $\mathcal{M}_{2n}^\times$ with the same norm can be connected by the action of an element of $SU(n)$. Hence,
			the special unitary action $\varphi:SU(n)\times \mathcal{M}_{2n}^\times\rightarrow \mathcal{M}_{2n}^\times$, with $n>1$, has $(2n-1)$-dimensional orbits, which are embedded submanifolds of $\mathcal{M}_{2n}^\times$ because $SU(n)$ is compact. Since $\dim\,SU(n)=n^2-1\geq 2n-1$ for $n>1$, we can choose around any point of $\mathcal{M}_{2n}^\times$ an open neighbourhood $A_0$ where $2n-1$ fundamental vector fields of $\varphi$ are linearly independent at each point. As they are also Hamiltonian vector fields, their Hamiltonian functions, which can be taken of the form {$f_\alpha (\psi) := \frac 12 \langle \psi, H_\alpha\psi\rangle$}  with $H_\alpha\in \mathfrak{su}(n)$, $\phi\ in\mathcal{M}_{2n}^\times$, and $\alpha\in\overline{1,2n-1}$,  are  functionally independent on $A_0$. These functions are invariant under the action $\varphi_{U(1)}:(e^{{\rm i}\varphi},\psi)\in U(1)\times \mathcal{M}^\times_{2n}\mapsto e^{{\rm i}\varphi}\psi\in \mathcal{M}_{2n}^\times$ and give rise to well-defined functions $f_1|_{\mathcal{R}_n},\ldots,f_{2n-1}|_{\mathcal{R}_n},$ on an open subset of ${\mathcal{R}_n}$. As $f_1,\ldots, f_{2n-1}$ are constant on the leaves of $\varphi_{U(1)}$ and functionally independent on $A_0 \subset \mathcal{M}_{2n}^\times$, then $f_1|_{\mathcal{R}_n},\ldots,f_{2n-1}|_{\mathcal{R}_n}$ are functionally independent and provide a local coordinate system on ${\mathcal{R}_n}$.
		\end{proof}
		
		\begin{example}\label{Coor}
			A two-level system is described by a four-dimensional manifold $\mathcal{M}_4$ with coordinates $(q_1, p_1, q_2, p_2)$ given in (\ref{CoorH}). Its dynamics is described by a $t$-dependent vector field $X=\sum_{\alpha=1}^{3}B_\alpha(t)X_\alpha$, where
			the vector fields $X_1, X_2, X_3$ are given in \eqref{2levelXj} and admit Hamiltonian functions
			$$
			h_1(\psi) =\frac 12(q_1q_2 + p_1p_2),\quad h_2(\psi) =\frac 12 (q_1p_2-q_2p_1),\quad h_3(\psi) = \frac 14(q_1^2+p_1^2-q_2^2-p_2^2),\quad \forall\psi\in \mathcal{M}_4
			$$ relative to the natural symplectic structure on $\mathcal{M}_4$ appearing in \eqref{eq:hj}.
			As stated in Lemma \ref{CoorHU}, these functions define a coordinate system $\{h_1,h_2,h_3\}$ on ${\mathcal{R}_{2}}$. To verify it, let us consider the projection $\pi_\mathcal{MR}:\mathcal{M}^\times_4\rightarrow  \mathbb{R}^3_0$ by:
			\begin{equation}
			\label{eq:defPi}
			\pi_\mathcal{MR} (\psi) := (x:=h_1(\psi),y:=h_2(\psi),z:=h_3(\psi)),
			\end{equation}
			and show that ${\mathcal{R}_{2}}\simeq \mathbb{R}^3_0$.
			Since  {$x^2+y^2+z^2=\frac{1}{16} \langle \psi, \psi \rangle^2$}  and $(0,0)\notin \mathcal{M}_4^\times$, the image of $\pi_\mathcal{MR}$ does not contain the origin $(0,0,0)$ and $\pi_\mathcal{MR}$ takes values in $\mathbb{R}^3_0$ as assumed.
			
			 {Coming back to}  the complex notation of $\Hil=\mathbb{C}^2$,  we write $\psi = (z_1, z_2) = (q_1 + {\rm i}p_1, q_2 + {\rm i}p_2)$. Therefore, 
			$$
			x(\psi)= \frac 12\mathfrak{Re}\langle z_1,z_2\rangle,\qquad y(\psi)= \frac 12\mathfrak{Im}\langle z_1,z_2\rangle,\qquad z(\psi)=\frac 14 (|z_1|^2-|z_2|^2).
			$$
			for every $\psi = (z_1,z_2)\in \mathcal{M}_4^\times$. Hence, $x,y,z$ are constant along the equivalence classes of ${\mathcal{R}_{2}}$ and if $\psi,\hat\psi\in \mathcal{M}_4^\times$ belong to the same equivalence class of $\mathcal{R}_{2}$, then $\pi_\mathcal{MR} (\psi)= \pi_\mathcal{MR} (\hat \psi)$. 
			Let us additionally show that if $\pi_\mathcal{MR} (\psi)=\pi_\mathcal{MR} (\hat \psi)$, then 
			$\psi=(z_1,z_2)$ and $\hat\psi:=(\hat z_1,\hat z_2)$ belong to the same equivalence class. Indeed, if $\pi_\mathcal{MR} (z_1,z_2)=\pi_\mathcal{MR} (\hat z_1,\hat  z_2)$, then $\sqrt{x^2+y^2+z^2}=(|z_1|^2+|z_2|^2)/4=(|\hat z_1|^2+|\hat z_2|^2)/4$ and, since $z(\psi)=z(\hat \psi)$, we obtain $|z_i|=|\hat z_i|$ for $i=1,2$. Therefore, $z_j=e^{{\rm i}\varphi_j}\hat z_j$ for certain $\varphi_j\in \mathbb{R}$  with $j=1,2$. In view of this and $\langle z_1,z_2\rangle=\langle \hat z_1,\hat z_2\rangle,$ we obtain $\varphi_1-\varphi_2=2\pi k$, for $k\in \mathbb{Z}$ and hence $(z_1,z_2)=e^{{\rm i}\varphi_1}(\hat z_1,\hat z_2)$. Thus, if $\pi^{-1}_\mathcal{MR} (x,y,z)$ is not empty, it gives rise to {an} equivalence class of $\mathcal{R}_2$.
			
			Let us prove that $\pi_\mathcal{MR}$ is a surjection.  For every $(x,y,z)\in \mathbb{R}_0^3$, we can prove that 
			$$
			\pi_\mathcal{MR} \left( \left[2(\sqrt{x^2+y^2+z^2}+z) \right] ^{1/2}, \left[ 2(\sqrt{x^2+y^2+z^2}-z) \right]^{1/2}e^{{\rm i}\Theta}\right)=(x,y,z),
			$$
			where,  the angle $\Theta\in [0,2\pi )$ satisfies for $x^2+y^2\neq 0$ the relation
			$$
			\frac{x}{\sqrt{x^2+y^2}}=\cos \Theta,\qquad \frac{y}{\sqrt{x^2+y^2}}=\sin \Theta
			$$
			and it is arbitrary for $x^2+y^2=0$.
			The above  {expressions show} that $\pi_\mathcal{MR}$  {is surjective}. Therefore, $\pi_\mathcal{MR}^{-1}(x,y,z)$ is the equivalence class of an element of $\mathcal{R}_2$ for every $(x,y,z)\in \mathcal{R}_2$ and $\mathcal{R}_2\simeq \mathbb{R}^3_0$.  
			

		\end{example}

		\begin{proposition}\label{ProjU(1)} The $t$-dependent Schr\"odinger equation (\ref{SchUni}), when restricted to $\mathcal{M}_{2n}^\times$, can be projected onto $\mathcal{R}_n$ originating a Lie--system $X_{\mathcal{R}_n}$ possessing a VG--Lie algebra $V_{\mathcal{R}_n}\simeq \mathfrak{su}(n)$ of Hamiltonian vector fields with respect to the projection of $\Lambda^\times$ on $\mathcal{M}_{2n}^\times$ {onto}  $\mathcal{R}_n$.
		\end{proposition}
		\begin{proof} The $\mathbb{C}$-linear Lie group action $\varphi_{\mathcal{M}^\times_{2n}}:U(n)\times \mathcal{M}_{2n}^\times\rightarrow \mathcal{M}_{2n}^\times$ induces, due to its $\mathbb{C}$-linearity, {another action on   
		 $\mathcal{R}_n$ such that the map $\pi_{\mathcal{MR}}$ is equivariant, as follows:}
			$$
			\begin{array}{rcccc}
			\varphi_{\mathcal{R}_n}:&U(n)\times \mathcal{R}_n&\rightarrow &\mathcal{R}_n,\\
			& (g,[\psi]_\mathcal{R})&\mapsto& [\varphi_{\mathcal{M}_{2n}^\times}(g,\psi)]_\mathcal{R}.\\
			\end{array}
			$$
			As a consequence, the fundamental vector fields of $V_{\mathcal{M}^\times_{2n}}$ project onto $\mathcal{R}_n$ giving rise to a new finite-dimensional Lie algebra of vector fields $V_{\mathcal{R}_n}$ and the {projection}  map $\pi_\mathcal{MR}:\mathcal{M}_{2n}^\times\rightarrow \mathcal{R}_n$ induces a Lie algebra morphism $\pi_{\mathcal{MR}*}|_{V_{\mathcal{M}^\times_{2n}}}:V_{\mathcal{M}^\times_{2n}}\rightarrow V_{\mathcal{R}_n}$. Then, the restriction to $\mathcal{M}_{2n}^\times$ of the Schr\"odinger equation on (\ref{SchUni}) also projects onto $\mathcal{R}_n$ giving rise to a system $X_{\mathcal{R}_n}$.
			
			Let us prove that  $X_{\mathcal{R}_n}$ admits a VG--Lie algebra isomorphic to $\mathfrak{su}(n)$.  As $V_{\mathcal{M}^\times_{2n}}\simeq \mathfrak{u}(n)\simeq \mathbb{R}\oplus \mathfrak{su}(n)$, the kernel of $\pi_{\mathcal{MR}*}|_{V_{\mathcal{M}^\times_{2n}}}$, which is an ideal of $V_{\mathcal{M}^\times_{2n}}$, may be zero, isomorphic to $\mathbb{R}$, to $\mathfrak{su}(n)$ or to $\mathfrak{u}(n)$. The one-parameter group of diffeomorphisms induced by the vector field $\Gamma^\times$ is given by $F_t:\psi\in\mathcal{M}^\times_{2n}\mapsto e^{{\rm i}t}\psi\in \mathcal{M}_{2n}^\times$. Hence, $\pi_{\mathcal{MR}*}\Gamma^\times=0$ and $\Gamma^\times$ belongs to the center of $V_{\mathcal{M}^\times_{2n}}$, i.e.
 {$\Gamma^\times\in \mathfrak{z}(V_{\mathcal{M}^\times_{2n}})\simeq \mathbb{R}$}. If $n=1$, then this shows that ${\rm Im}\,\pi_{\mathcal{MR}*}|_{V^\times_{\mathcal{M}_{2}}}=\{0\}\simeq \mathfrak{su}(1)$ and the result follows. Meanwhile, $V_{\mathcal{R}_n}\neq 0$ for $n>1$ and in view of the decomposition $V_{\mathcal{M}^\times_{2n}}\simeq \mathbb{R}\oplus\mathfrak{su}(n)$, we get that  $\ker\pi_{\mathcal{MR}*}\simeq \langle \Gamma^\times\rangle$ and ${\rm Im} \,\pi_{\mathcal{MR}*}|_{V^\times_{\mathcal{M}_{2n}}}\simeq \mathfrak{su}(n)$. Thus, the projection of (\ref{SchUni}) onto $\mathcal{R}_n$ admits a VG--Lie algebra $V_{\mathcal{R}_n}\simeq \mathfrak{su}(n)$.
		\end{proof}

		\begin{example}\label{ExHU} A simple computation shows that
			there exist vector fields $Y_\alpha$ on $\mathcal{R}_2$ such that
			$\pi_{\mathcal{MR}*} (X_{\alpha}) = Y_{\alpha}$ for $\alpha=1,2,3$. Indeed,
			\begin{equation}\label{Vec}
			Y_1=-z\frac{\partial}{\partial y}+y\frac{\partial}{\partial z},\qquad Y_2=z\frac{\partial}{\partial x}-x\frac{\partial}{\partial z},\qquad Y_3=-y\frac{\partial}{\partial x}+x\frac{\partial}{\partial y}.
			\end{equation}
			The Lie brackets between these  vector fields read
			$$
			[Y_1, Y_2] = -Y_3,\qquad  [Y_2, Y_3] = -Y_1,\qquad [Y_3, Y_1] = -Y_2,
			$$
			that is 	$[Y_j, Y_k] = -\sum_{l=1}^3\epsilon_{jkl} Y_l$ for $j,k,l=1,2,3$. The projection of $X=\sum_{\alpha=1}^3B_\alpha(t)X_\alpha$ to $\mathcal{R}_2$, i.e. the $t$-dependent vector field $X_{\mathcal{R}_2}$ satisfying $(X_{\mathcal{R}_2})_t=\pi_{\mathcal{MR}*}X_t$, becomes
			\begin{equation}\label{YField}
			X_{\mathcal{R}_2}= \sum_{\alpha=1}^3B_\alpha(t)Y_\alpha.
			\end{equation}
			
			This is exactly the same relation given in \eqref{LiePSigma}, which shows that $\mathfrak{Y}:=\langle Y_1,Y_2,Y_3\rangle\simeq \mathfrak{su}^*(2)$.  Therefore, $X_{\mathcal{R}_2}$ is a Lie system. Observe that $Y_1,Y_2,Y_2$ span a two-dimensional distribution $\mathcal{D}^\mathfrak{Y}$.

			
		\end{example}
		
		The following proposition shows that $\mathcal{R}_{n}$ can be endowed with a Poisson structure  turning $V_{\mathcal{R}_n}$ into a Lie algebra of Hamiltonian vector fields.
		
		\begin{proposition} The system $X_{\mathcal{R}_{n}}$ is a Lie--Hamilton system with respect to the Poisson bivector $\pi_{\mathcal{MR}*}\Lambda^\times$.
		\end{proposition}
		\begin{proof} Since $\mathcal{L}_{\Gamma^\times} \Lambda^\times=0$, the Poisson bivector $\Lambda^\times$ on $\mathcal{M}_{2n}^\times$ can be projected onto $\mathcal{R}_n$.   Additionally, $
			\pi_{\mathcal{MR}*}[\Lambda,^\times\Lambda^\times]_{SN}=[\pi_{\mathcal{MR}*}\Lambda^\times, \pi_{\mathcal{MR}*}\Lambda^\times]_{SN},$
			where $[\cdot,\cdot]_{SN}$ is the Schouten-Nijenhuis bracket \cite{IV}. So, $\pi_{\mathcal{MR}*}\Lambda^\times$ is a Poisson bivector on $\mathcal{R}_n$. The vector fields $X_\alpha$ spanning the VG--Lie algebra $V_{\mathcal{M}^\times_{2n}}$ for $X_{\mathcal{M}_{2n}^\times}$ are Hamiltonian relative to the restrictions to $\mathcal{M}_{2n}^\times$ of the functions $h_\alpha$ in \eqref{eq:hj}. Such Hamiltonian functions are invariant relative to the action of $U(1)$ on $\mathcal{M}^\times_{2n}$ and hence projectable onto $\mathcal{R}_n$. The projections $\pi_{\mathcal{MR}*}X_\alpha$ are also Hamiltonian vector fields with Hamiltonian functions $x_\alpha=\pi^*_\mathcal{MR} (h^\times_\alpha)$. Therefore, the VG--Lie algebra $V_{\mathcal{R}_n}$ on $\mathcal{R}_n$ consists of Hamiltonian vector fields relative to $\pi_{\mathcal{MR} *}\Lambda^\times$.
		\end{proof}
		
		\begin{example}
			
			The Poisson bivector $\Lambda^\times$ on $\mathcal{M}_{4}^\times$ projects onto $\mathcal{R}_2$ giving rise to the Poisson bivector 
			\begin{equation}
			\label{eq:2LambdaHat}
			\widehat{\Lambda} =  z \frac{\partial}{\partial x} \wedge \frac{\partial}{\partial y}
			+ x \frac{\partial}{\partial y} \wedge \frac{\partial}{\partial z}
			+ y \frac{\partial}{\partial z} \wedge \frac{\partial}{\partial x}
			\end{equation}
			in the coordinate system given in Example \ref{Coor}.
		\end{example}

		\begin{proposition} The system $X_{\mathcal{R}_n}$ consists of Killing vector fields with respect to the metric induced by the projection of the tensor field $G^\times$ to $\mathcal{R}_n$.
		\end{proposition}
		\begin{proof} The Lie derivative of $G^\times$ with respect to $\Gamma^\times$ is zero. Hence, $G^\times$ projects onto $\mathcal{R}_n$.  Since $G^\times$ is Riemannian, it is non-degenerate. So is its projection onto $\mathcal{R}_n$ giving rise to a Riemannian metric on $\mathcal{R}_n$. The vector fields of $V_{\mathcal{M}^\times_{2n}}$ are Killing relative to $G^\times$ and projectable under $\pi_{\mathcal{MR}*}$. Therefore, their projections, namely the elements of $V_{\mathcal{R}_n}$, are also Killing vector fields relative to the projection of $G^\times$ onto $\mathcal{R}_n$ and span a VG--Lie algebra $V_{\mathcal{R}_n}$ of Killing vector fields.
		\end{proof}
		\begin{example}
			The tensor field $G^\times$ on $\mathcal{M}_4^\times$ projects onto $\mathcal{R}_2$ giving rise to the tensor field
			\begin{equation}
			\widehat{G} = (x^2+y^2 + z^2)^{1/2} \left[\frac{\partial}{\partial x} \otimes \frac{\partial}{\partial x}+
			\frac{\partial}{\partial y} \otimes \frac{\partial}{\partial y}
			+ \frac{\partial}{\partial z} \otimes \frac{\partial}{\partial z}
			\right]. \\
			\end{equation}
			It naturally allows us to define a Riemannian metric $\widehat g$ on $\mathcal{R}_2$ given by
			\begin{equation}
			\label{eq:2gHat}
			\widehat{g}= \frac{{\rm d}x\otimes {\rm d}x+{\rm d}y\otimes {\rm d}y+{\rm d}z\otimes {\rm d}z}{(x^2+y^2+z^2)^{1/2}}.
			\end{equation}
			It is immediate to check that the Lie derivatives of $\hat g$ and $\hat G$ with respect to the vector fields (\ref{Vec}) are zero as stated in the previous proposition.
		\end{example}
		
			\section{Lie systems and projective Schr\"odinger equations}
		Let $\widehat H(t)$ be a $t$-dependent Hermitian operator on the Hilbert space $L^2(\mathbb{R}^3)$ of equivalence classes of square integrable measurable functions on $\mathbb{R}^3$ with respect to Lebesgue measure. Solutions to the $t$-dependent Sch\"odinger equation determined by $\widehat H(t)$ differing in a proportional non-zero $t$-dependent complex factor describe the same physical state. Hence, it is natural to consider whether $t$-dependent Schr\"odinger equations admit such a symmetry. Nevertheless, (\ref{SchUni}) is not invariant under the change of phase $ \psi\mapsto f(t)\psi(t)$ for a non-zero complex valued function $f$. There is 
		however another differential equation, the so-called {\it projective Schr\"odinger equation}, that is invariant under such a change of phase while admitting the particular solutions to the original $t$-dependent Schr\"odinger equation \cite{BH01}. It is given by
		\begin{equation}\label{ProSchEqu}
		{\rm i} \left[\psi_{\bf y}\frac{\d\psi_{\bf x}}{\d t}-\psi_{\bf x}\frac{\d \psi_{\bf y}}{\d t}\right]=\psi_{\bf y}\widehat H_x\psi_{\bf x}-\psi_{\bf x}\widehat H_y\psi_{\bf y}, \quad \psi:\mathbb{R}^3\rightarrow \mathbb{C},
		\end{equation}
with $\psi_{\bf x}:=\psi({\bf x})$, $\psi_{\bf y}:=\psi({\bf y})$ and ${\bf x},{\bf y}\in \mathbb{R}^3$. These differential equations can be further generalised while admitting the previous symmetry and covering particular solutions to $t$-dependent Schr\"odinger equations on (probably infinite-dimensional Hilbert spaces) by writing 
		\begin{equation}\label{ProSchEquH}
		{\rm i} \left[{\rm Id}\otimes \frac{{\rm d}}{{\rm d}t}-\frac{{\rm d}}{{\rm d}t}\otimes {\rm Id}\right]\psi\otimes \psi= \left[{\rm Id}\otimes \widehat H(t)-\widehat H(t)\otimes {\rm Id}\right]\psi\otimes \psi,\qquad \forall \psi\in \mathcal{H},
		\end{equation}
		where $\psi \otimes\psi \in \mathcal{H}\otimes\mathcal{H}$ and $A\otimes B$ are the tensorial products of the Hermitian operators $A,B:\mathcal{H}\rightarrow\mathcal{H}$. It said that $\psi\in \mathcal{H}$ is a solution to (\ref{ProSchEquH}) is $\psi\otimes \psi$ satisfies it. 
		
		Nevertheless, (\ref{ProSchEqu}) and (\ref{ProSchEquH}) have not a clear geometric interpretation. Equation (\ref{ProSchEqu}) depends on the values of the function $\psi$ in two different points and $(\ref{ProSchEquH})$ is defined on tensorial products of the form $\psi \otimes \psi$. Instead, we will make use of a projection of the $t$-dependent Schr\"odinger equation on $\mathcal{H}_0$ onto $\P_n$ to recover its solutions up to a global $t$-dependent change of phase and, therefore, recovering the same solutions of (\ref{ProSchEquH}).

		\begin{lemma} The $t$-dependent vector field $X$ on $\mathcal{M}^\times_{2n}$ related to a $t$-dependent Schr\"odinger equation 
			\begin{equation}\label{equ:tSch}
			\frac{\d \psi}{\d t}= -{\rm i}H(t)\psi,\qquad -{\rm i}H(t)\in \mathfrak{u}(n), \quad \psi\in\mathcal{H}_0,
			\end{equation}
is projectable under the fibration $\pi_\mathcal{MP}: \mathcal{M}_{2n}^\times \rightarrow \mathcal{P}_n$.
		\end{lemma}
		\begin{proof} Consider the natural actions of the  unitary group $U(n)$ on $\mathcal{M}^\times_{2n}$ and on the projective space $\mathcal{P}_n$ given by
		$$
		\varphi:(U,\psi)\in U(n)\times\mathcal{M}^\times_{2n}\mapsto U\psi \in\mathcal{M}^\times_{2n},\quad \varphi_{\mathcal{P}_n}:(U,[\psi]_\mathcal{P})\in U(n)\times\mathcal{P}_n\mapsto [U\psi]_\mathcal{P}\in\mathcal{P}_n,
		$$
respectively. Then, the map   $\pi_{\mathcal{MP}}:\mathcal{M}_{2n}^\times\rightarrow \mathcal{P}_{n}$                                                  
is equivariant. 
By the definition of (\ref{equ:tSch}), the $t$-dependent vector field $X$ belongs, at each $t\in\mathbb{R}$, to the Lie algebra, $V_{\varphi_{\mathcal{P}_n}}$ of fundamental vector fields of $\varphi$. As $\pi_\mathcal{MP}$ is equivariant, the fundamental vector fields 
of the actions $\varphi$ and $\varphi_{\mathcal{P}_n}$  are  $\pi_\mathcal{MP}$-related and each vector field of $V_{\mathcal{M}^\times_{2n}}$ projects onto a fundamental vector field of $\varphi_{\mathcal{P}_n}$ and vice versa. This ensures $\pi_{\mathcal{MP}*}X$ to exist and to admit a VG--Lie algebra $V_{\mathcal{P}_n}:=V_{\varphi_{\mathcal{P}_n}}$.
			
		\end{proof}
		
		\begin{definition} Given a  $t$-dependent Schr\"odinger equation $X$ of the form (\ref{equ:tSch}),  we call {\it projective Schr\"odinger equation} on $\mathcal{P}_{n}$ the system of differential equations
			\begin{equation}\label{SCPH}
			\frac{{\rm d} \xi}{{\rm d}t}={X}_{\mathcal{P}_n}(t,\xi),\qquad \xi\in \mathcal{P}_{n},\qquad \forall t\in\mathbb{R},
			\end{equation}
			where $X_{\mathcal{P}_n}$ is the projection onto $\mathcal{P}_{n}$ of $X$ relative to  
			the projection $\pi_{\mathcal{MP}}:\mathcal{M}_{2n}^\times\rightarrow \mathcal{P}_{n}$.
		\end{definition}
		
		\begin{proposition} A non-vanishing curve $\psi:\mathbb{R}\rightarrow \mathcal{M}_{2n}^\times$ is a particular solution to the restriction of the projective Schr\"odinger equation  (\ref{ProSchEquH}) to $\mathcal{H}_0$ if and only if $\pi_{\mathcal{MP}}\circ \psi$  is a particular solution to (\ref{SCPH}).
		\end{proposition}
		\begin{proof} The projective Schr\"odinger equation (\ref{ProSchEquH}) can be brought into the form
		$
		[\widehat E(t) \otimes {\rm Id}]( \psi \otimes \psi )=[{\rm Id}\otimes \widehat E(t) ](\psi \otimes \psi) ,
		$
		where $\widehat E(t):=\partial_t-{\rm i}\widehat H(t)$. Let $\psi(t)$ be a particular solution to (\ref{ProSchEquH}). Then, $\widehat E(t)\psi(t)=g(t)\psi(t) $ for a certain $t$-dependent function $g(t)$ and there always exists a $t$-dependent function $f(t)$ such that $\widehat E(t)[f(t)\psi(t)]=0$ and $\psi_S(t):=f(t)\psi(t)$ becomes a solution to the standard $t$-dependent Schr\"odinger equation. Since the Schr\"odinger equation is related to a $t$-dependent Hermitian Hamiltonian operator, it follows that $\psi_S(t)$ has a constant module. By assumption $\psi(t)$ does not vanishes and therefore $\psi_S(t)$ does not vanishes neither.  Hence, both curves can be projected onto $\mathcal{P}_n$ and $\pi_\mathcal{MP}(\psi(t))=\pi_\mathcal{MP}(\psi_S(t))$. Since $X_{\mathcal{M}_{2n}^\times}$ projects onto ${X}_{\mathcal{P}_n}$, then $\pi_\mathcal{MP}(\psi_S(t))$ is a particular solution to (\ref{SCPH}) and $\psi(t)$ projects onto a particular solution to (\ref{SCPH}). 
			
			Conversely, if $\psi(t)$ projects onto a solution $\psi_\mathcal{P}(t)$ to (\ref{SCPH}), there exists a solution $\psi_S(t)$ to the Schr\"odinger equation projecting onto $\psi_\mathcal{P}(t)$. Since $\psi(t)$ and $\psi_S(t)$ project onto $\psi_\mathcal{P}(t)$, they differ at each $t$ on a phase and $\psi(t)=\lambda(t)\psi_S(t)$ for a certain complex function $\lambda(t)$. Note that the fact that $\psi(t)$ projects onto $\psi_\mathcal{P}(t)$ implies that it does not vanish. Hence, $\psi(t)$ is a non-vanishing particular solution to (\ref{ProSchEqu}).
		\end{proof}

		\begin{theorem} The system (\ref{SCPH}) is a Lie system related to a VG--Lie algebra $V_{\mathcal{P}_n}\simeq \mathfrak{su}(n)$ consisting of K\"ahler vector fields with respect to the natural K\"ahler structure on $\mathcal{P}_n$.
		\end{theorem}
		\begin{proof}
			In view of Theorem \ref{KahH}, the vector fields of $V_{\mathcal{M}^\times_{2n}}$ for (\ref{equ:tSch}) leave invariant $G^\times$ and $\Lambda^\times$, i.e. $\mathcal{L}_X\Lambda^\times=\mathcal{L}_XG^\times=0$ for every $X\in V_{\mathcal{M}_{2n}}$. Since $[\Gamma^\times,X] = [\Delta^\times, X] =0$, $\Gamma^\times f^\times_I(\psi)=\Delta^\times f_I^\times(\psi)=0$ and in view of the expressions (\ref{eq:13}), it follows that 
			\begin{equation}\label{com}
			\mathcal{L}_X\Lambda^\times_\mathcal{P}=\mathcal{L}_XG^\times_\mathcal{P}=0,
			\end{equation} 
			where $G_\mathcal{P}$ and $\Lambda_\mathcal{P}$ are given by (\ref{eq:13}). The tensor fields $\Lambda_\mathcal{P}$ and $G_\P$ on $\mathcal{P}_n$ generate the canonical K\"ahler structure on $\mathcal{P}_n$ (see \cite{CM08} for details). Since the vector fields of $V_{\mathcal{M}^\times_{2n}}$ project onto $\mathcal{P}_n$ and in view of (\ref{com}), it turns out that the projection {onto} $\mathcal{P}_n$ of the vector fields of $V_{\mathcal{M}^\times_{2n}}$ span a VG--Lie algebra $V_{\mathcal{P}_n}$ for (\ref{SCPH}) of K\"ahler vector fields relative to the natural K\"ahler structure on $\mathcal{P}_n$.  
			
			Let us prove that $V_{\mathcal{P}_n}\simeq\mathfrak{su}(n)$. The {natural projection} map $\pi_\mathcal{MP}:\mathcal{M}_{2n}^\times \rightarrow \mathcal{P}_n$ induces a Lie algebra morphism $\pi_{\mathcal{MP}*}|_{V_{\mathcal{M}^\times_{2n}}}:V_{\mathcal{M}^\times_{2n}}\rightarrow V_{\mathcal{P}_n}$. As $V_{\mathcal{M}^\times_{2n}}\simeq \mathfrak{u}(n)\simeq \mathbb{R}\oplus \mathfrak{su}(n)$, the kernel of $\pi_{\mathcal{MP}*}|_{V_{\mathcal{M}^\times_{2n}}}$, which is an ideal of $V_{\mathcal{M}^\times_{2n}}$, may be zero, {either} isomorphic to $\mathbb{R}$, to $\mathfrak{su}(n)$ or to $\mathfrak{u}(n)$. The one-parameter group of diffeomorphism induced by the vector field $\Gamma$ on $\mathcal{M}^\times_{2n}$ is given by $F_t:\psi\in\mathcal{M}^\times_{2n}\mapsto e^{{\rm i}t}\psi\in\mathcal{M}^\times_{2n}$. Hence, $\pi_{\mathcal{MP}*}\Gamma^\times=0$ and $\Gamma^\times\in \mathfrak{z}(V_{\mathcal{M}^\times_{2n}})\simeq \mathbb{R}$. Since $V_{\mathcal{P}_n}\neq \{0\}$ and in view of the decomposition of $V_{\mathcal{M}^\times_{2n}}$, we get that  $\ker\pi_{\mathcal{MP}*}\simeq \langle \Gamma^\times\rangle$ and ${\rm Im} \,\pi_{\mathcal{MP}*}|_{V_{\mathcal{M}^\times_{2n}}}\simeq \mathfrak{su}(n)$. Thus, $V_{\mathcal{P}_n}\simeq \mathfrak{su}(n)$.
			
		\end{proof} 
		
		The following proposition will be helpful to study projective Schr\"odinger equations as a restriction of a system on $\mathcal{R}_n$.\\
		
		\begin{minipage}{6cm}
			\xymatrix{
				&\underset{}{\mathcal{M}_{2n}^\times}\ar[dd]_{\pi_{\mathcal{MP}}}
				\ar[ld]_{\pi_{\mathcal{MR}}}&\\
				\underset{}{\mathcal{R}_n\ar[dr]}\ar[dr]_{\pi_{\mathcal{RP}}}&&\underset{}{\mathcal{S}_n}\ar[ul]_{\iota_\mathcal{S}}\ar[dl]\ar@/_{4mm}/@{->}[ll]_{\quad\quad\pi_{\mathcal{S}\mathcal{R}}}\ar[dl]^{\pi_{\mathcal{SP}}}\\
				&\underset{}{{\mathcal{P}_n}}\ar@/^{8mm}/[ul]^{\iota_{\mathcal{P}}}&}
		\end{minipage}
		\begin{minipage}{11.7cm}
			\begin{proposition}\label{RnPn}Let $\pi_{\mathcal{SR}}:\mathcal{S}_n\rightarrow \mathcal{R}_n$ be of the form $\pi_{\mathcal{SR}} :=\pi_{\mathcal{MR}}\circ \iota_{\mathcal{S}}$, with $\iota_{\mathcal{S}}: \mathcal{S}_n\rightarrow \mathcal{M}_{2n}^\times$ being the natural embedding of $\mathcal{S}_n$ in $\mathcal{M}_{2n}^\times$, and let $\pi_{\mathcal{R}\mathcal{P}}:[\psi]_\mathcal{R}\in \mathcal{R}_n\mapsto  [\psi]_\mathcal{P}\in \mathcal{P}_n$. 
				There exists a differentiable embedding $\iota_{\mathcal{P}}: \mathcal{P}_n\rightarrow \mathcal{R}_n$ such that $\pi_{\mathcal{RP}} \circ \iota_{\mathcal{P}} = {\rm Id}_{\mathcal{P}_n}$ and $\iota_{\mathcal{P}*} X_{\mathcal{P}_n} = \pi_{\mathcal{SR}*} X_{\mathcal{S}_n}$.
			\end{proposition}
					\end{minipage}
			\begin{proof} It is immediate that the diagram aside is commutative. For every element of $\mathcal{P}_n$ there exists an element of $\mathcal{S}_n\subset \mathcal{M}_{2n}^\times$ projecting to it under $\pi_{\mathcal{SP}}$. Hence, $\pi_{\mathcal{RP}} \circ \pi_{\mathcal{SR}} (\mathcal{S}_n) = \mathcal{P}_n$ and $\pi_{\mathcal{RP}}|_{\pi_{\mathcal{SR}}(\mathcal{S}_n)}: \pi_{\mathcal{SR}} (\mathcal{S}_n) \rightarrow \mathcal{P}_n$ is surjective. Let us show that it is also injective. Let $[\psi]_\mathcal{R}$, with $\psi\in \mathcal{M}_{2n}^\times$, be the equivalence class of functions on $\mathcal{M}_{2n}^\times$ differing in a complex-phase of module one. If $\pi_{\mathcal{RP}} ([\psi_1]_\mathcal{R}) = \pi_{\mathcal{RP}} ([\psi_2]_\mathcal{R})$ and $\|\psi_1\|=\|\psi_2\|$, then $\psi_1 = e^{{\rm i}\phi} \psi_2$ with $\phi\in \mathbb{R}$. Hence, $[\psi_1]_\mathcal{R} = [\psi_2]_\mathcal{R}$ and $\pi_{\mathcal{RP}}$ is a bijection when restricted to $\pi_{\mathcal{SR}} (\mathcal{S}_n)$ and the inverse map is defined to be $\iota_{\mathcal{P}}$. From this it is immediate that $\iota_{\mathcal{P}*} X_{\mathcal{P}_n} = \pi_{\mathcal{SR}*} X_{\mathcal{S}_n}$.
				
			\end{proof}

			\section{Superposition rules for special unitary Schr\"odinger equations}
			We now prove that, apart from the very well-known standard linear superposition rules, $t$-dependent Schr\"odinger equations  associated with $t$-dependent traceless Hermitian Hamiltonian operators have other nonlinear ones which depend, generally, on fewer solutions. Superposition rules for different types of projections of Schr\"odinger equations are investigated.

			\begin{theorem}\label{MT1}
				Every Schr\"odinger equation on $\mathcal{M}_{2n}$, with $n>1$, related to a VG--Lie algebra $V\subset V_{\mathcal{M}_{2n}}$ isomorphic to $\mathfrak{su}(n)$ admits a  superposition rule depending on $n-1$ particular solutions.
			\end{theorem}
			\begin{proof} In light of Theorem \ref{KahH}, our Schr\"odinger equation admits a VG--Lie algebra $V$ of K\"ahler vector fields isomorphic to $\mathfrak{su}(n)$.
				To derive a superposition rule, we determine the smallest $m\in \mathbb{N}$ so that the diagonal prolongations to $\mathcal{M}_{2n}^m$ of the vector fields of $V$  span a a distribution of rank $\dim V$ at a generic point. Since $V\subset V_{\mathcal{M}_{2n}}$ and $V\simeq \mathfrak{su}(n)$, the elements of $V$ are fundamental vector fields of the standard linear action of $SU(n)$ on $\mathcal{M}_{2n}$ (thought of as a $\mathbb{C}$-linear space). The diagonal prolongations of $V$ to $\mathcal{M}_{2n}^m$ span the tangent space to the orbits of the Lie group action
				$$
				\begin{array}{rccc}
				\varphi^m:&SU(n)\times \mathcal{M}_{2n}^m&\longrightarrow&\mathcal{M}_{2n}^m\\
				&(U;\psi_1,\ldots,\psi_m)&\longmapsto & (U\psi_1,\ldots, U\psi_m).
				\end{array}
				$$
				The fundamental vector fields of this action span a distribution of rank $\dim V$ at $\xi\in \mathcal{M}_{2n}^{m}$ if and only if its isotropy group $\mathcal{O}_\xi$ at $\xi$ is discrete. Let us set $m:=n-1$. The elements  $U\in \mathcal{O}_\xi$, with $\xi:=(\psi_1,\ldots,\psi_{n-1})\in \mathcal{M}_{2n}^{n-1}$, satisfy
				\begin{equation}\label{Conp}
				U\psi_j=\psi_j,\qquad j=\overline{1,n-1}.
				\end{equation}
				At a generic point of $\mathcal{M}_{2n}^{n-1}$, we can assume that $\psi_1,\ldots,\psi_{n-1}$ are linearly independent elements of $\mathbb{C}^n$ (over $\mathbb{C}$). Then, the knowledge of the action of $U$ on these elements fixes $U$ on $\langle \psi_1,\ldots,\psi_{n-1}\rangle_\mathbb{C}\subset \mathcal{M}_{2n}$, where it acts as the identity map. If $\psi$ is orthogonal to $\langle \psi_1,\ldots,\psi_{n-1}\rangle_\mathbb{C}$ with respect to the natural Hermitian product on $\mathbb{C}^n$, then  $U\psi$ must also be orthogonal to $\langle \psi_1,\ldots,\psi_{n-1}\rangle_\mathbb{C}$ because of (\ref{Conp}) and the unitarity of $U$. Therefore, $U\psi$ is proportional to $\psi$. Since $U\in SU(n)$, then $U\psi=\psi$ and $U={\rm Id}$. Therefore, the isotropy group of $\varphi^m$ is trivial at a generic point of $\mathcal{M}_{2n}^{n-1}$, the fundamental vector fields of $\varphi^m$ are linearly independent over $\mathbb{R}$ and there exists a superposition rule depending on $n-1$ particular solutions.
			\end{proof}
			
			\begin{note}
				It is worth noting that the isotropy group for $\varphi^m$ is not trivial at any point of  $\mathcal{M}_{2n}^m$ for $m<n-1$. Given $m$ linearly independent elements $\psi_1,\ldots,\psi_m\in \mathcal{M}_{2n}^\times$ over $\mathbb{C}$, we can construct several special unitary transformations on $\mathcal{M}_{2n}$ acting as the identity on $\langle \psi_1,\ldots,\psi_m\rangle_\mathbb{C}$ and leaving stable its orthogonal complement. Hence, the isotropy group on any point of $\mathcal{M}^m_{2n}$ is not discrete.
			\end{note}
			
			Since the elements of $U(n)$ act on $\mathcal{M}_{2n}$ preserving the norm relative to the standard Hermitian product on $\mathbb{C}^n$, the Lie group action $\varphi^m$  given in the proof of the previous theorem can be restricted to $\mathcal{S}^m_n$. In view of this, the previous proof can be slightly modified to prove that the restriction of $\varphi^m$ to $\mathcal{S}^m_n$ have a trivial isotropy group at a generic point for $m=n-1$ and $n>1$. As a consequence, we obtain the following corollary.
			\begin{corollary}\label{Cor1}
							Every Schr\"odinger equation $X_{\mathcal{S}_n}$, with $n>1$, related to a VG--Lie algebra $V\subset V_{\mathcal{S}_{n}}$ isomorphic to $\mathfrak{su}(n)$ admits a  superposition rule depending on $n-1$ particular solutions.
			\end{corollary}

			\begin{theorem}\label{MT1}
				Every Schr\"odinger equation $X_{\mathcal{R}_n}$, with $n>1$, admits a  superposition rule depending on $n$ particular solutions.
			\end{theorem}
			\begin{proof} In view of Proposition \ref{ProjU(1)},  the Schr\"odinger equation under study admits a VG--Lie algebra $V_{\mathcal{R}_n}$ of fundamental vector fields 
				isomorphic to $\mathfrak{su}(n)$. Also the proof of Proposition \ref{ProjU(1)} shows that the diagonal prolongation of the elements of $V_{\mathcal{R}_n}$ to $\mathcal{R}_n^m$ are the fundamental vector fields of the Lie group action 
				$$
				\begin{array}{rccc}
				\varphi_\mathcal{R}^m:&SU(n)\times \mathcal{R}_n^m&\longrightarrow&\mathcal{R}_n^m\\
				&(U;[\psi_1]_\mathcal{R},\ldots,[\psi_m]_\mathcal{R})&\longmapsto & ([U\psi_1]_\mathcal{R},\ldots, [U\psi_m]_\mathcal{R}).
				\end{array}
				$$
				To derive a superposition rule for $X_{\mathcal{R}_n}$, we determine the smallest $m\in \mathbb{N}$ so that the diagonal prolongations of a basis  $V_{\mathcal{R}_n}$ become linearly independent at a generic point. 
				This occurs at  $p\in \mathcal{R}_n^{m}$ if and only if the isotropy group of this action at $p$ is discrete. Let us set $m=n$. The elements of the isotropy group of $\varphi^n_\mathcal{R}$ at a generic point $p:=([\psi_1]_\mathcal{R},\ldots,[\psi_{n}]_\mathcal{R})\in \mathcal{R}_n^{n}$ satisfy
				\begin{equation}\label{Conp2}
				U[\psi_j]_\mathcal{R}=[\psi_j]_\mathcal{R},\qquad j\in\overline{1,n}.
				\end{equation}
				At a generic point of $\mathcal{R}_n^{n}$, we can assume that $\psi_1,\ldots,\psi_{n}$ are linearly independent elements of $\mathbb{C}^n$ (over $\mathbb{C}$). In view of (\ref{Conp2}), the operator $U$ diagonalises on the basis $\psi_1,\ldots,\psi_{n}$. Since $U\in U(n)$, we have that $\langle U\psi_i,U\psi_j\rangle=\langle \psi_i,\psi_j\rangle$ for $i,j=\overline{1,n}$ and all factors in the diagonal of the matrix representation of $U$ must be equal. As $U\in SU(n)$,  the multiplication of such diagonal elements must be equal to 1. This fixes $U=e^{{\rm i}2\pi k/n}$ for $k\in \mathbb{Z}$. Therefore, the stability group of $\varphi^n_{\mathcal{R}}$ is discrete at a generic point of $\mathcal{R}_n^{n}$, the fundamental vector fields of $\varphi_\mathcal{R}^n$ are linearly independent over $\mathbb{R}$ at a generic point and $X_{\mathcal{R}_n}$ admits  a superposition rule depending on $n$ particular solutions.
			\end{proof}
Recall that Proposition \ref{RnPn} states that $\mathcal{P}_n$ can be embedded naturally within $\mathcal{R}_n$. Additionally, the  projection $\pi_{\mathcal{R}\mathcal{P}}:\mathcal{R}_n\rightarrow \mathcal{P}_n$ is equivariant relative to the the Lie group action of $SU(n)$ on $\mathcal{R}_n$ and the action $\varphi_{\mathcal{P}}$ of $SU(n)$ on $\mathcal{P}_n$. Using these facts, we can easily prove the following corollary by using the same line of reasoning as in Corollary \ref{Cor1}.
	\begin{corollary}
		Every Schr\"odinger equation on $\mathcal{P}_n$, with $n>1$, related to a VG--Lie algebra $V_{\mathcal{P}_{n}}$ admits a  superposition rule depending on $n$ particular solutions.
	\end{corollary}

			The second interesting point is that the constants of motion needed to obtain a superposition rule for  special unitary Schr\"odinger equations can be obtained from the associated K\"ahler structure.
			
			%
			
			
\!\!\!\!
	\section{Superposition rules for one-qubit systems}
		In this section we illustrate our theory by describing superposition rules for one-qubit systems and their projections onto $\mathcal{S}_2, \mathcal{R}_2$ and $\mathcal{P}_2$. Observe that we can define the commutative diagram below. For the sake of completeness, we have added under each space the smallest number of particular solutions for its corresponding superposition rule.
		\medskip
		
		\begin{minipage}{4cm}
			
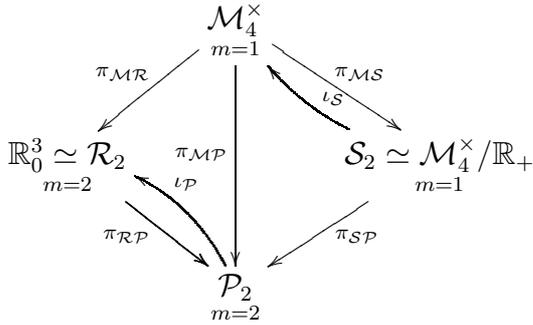
\xymatrix{
				&\underset{m=1}{\mathcal{M}_4^\times}\ar[dd]_{\pi_\mathcal{MP}}\ar[dr]^{\pi_{\mathcal{MS}}}\ar[ld]_{\pi_{\mathcal{MR}}}&\\
				\underset{m=2}{\mathbb{R}^3_0\simeq \mathcal{R}_2\ar[dr]}\ar[dr]_{\pi_{\mathcal{RP}}}&&\underset{m=1}{\mathcal{S}_2\simeq\mathcal{M}_4^\times/ \mathbb{R}_{+}}\ar@/^{4mm}/@{->}[ul]_{\iota_\mathcal{S}}\ar[dl]^{\pi_{\mathcal{SP}}}\\
				&\underset{m=2}{\mathcal{P}_2}\ar@/_{4mm}/@{->}[ul]_{\iota_\mathcal{P}}&}
		\end{minipage}
		\begin{minipage}{8.8cm}
			On each space we can define a Lie system admitting Vessiot--Guldberg Lie algebras of Hamiltonian vector fields relative to different compatible geometric structures, which in turn allows us to obtain their superposition rules geometrically.
			The following subsections provide these superposition rules, their relevant geometric properties and their potential applications in quantum mechanics. This will be carried out by applying our previous results. Our procedures will give rise to generalisations of our methods to systems with an arbitrary number of qubits.
		\end{minipage}
		
		\subsection{Superposition rule for a two-level system on $\mathcal{M}_4^\times$}
		
		Let us obtain a superposition rule for the system $X=\sum_{\alpha=1}^3B_\alpha(t)X_\alpha$ on $\mathcal{M}_4^\times$ given by (\ref{C2}) with $B_0(t)=0$. Our aim is to illustrate our previous theory while showing that there exists a superposition rule for $X$ depending just on one particular solution.
		
		It is an immediate consequence of Theorem \ref{KahH} that the restriction of system \eqref{C2} to $\mathcal{M}_4^\times$ is a K\"ahler--Lie system whose VG--Lie algebra $V=\langle X_1,X_2,X_3\rangle$, with $X_1,X_2,X_3$ given by (\ref{2levelXj}), consists of K\"ahler vector fields relative to the standard K\"ahler structure $(g,\omega, J)$ on $\mathcal{M}^\times_4$. Also, $X_t$ commutes with the phase change vector field $\Gamma$ and with the dilation vector field $\Delta$ for every $t\in\mathbb{R}$, namely $\Gamma$ and $\Delta$ are Lie symmetries of $X$.
		All these facts will be afterwards used to obtain superposition rules for (\ref{C2}) with $B_0=0$. 
		
		The number of particular solutions needed to obtain a superposition rule for the $t$-dependent vector field $X$ related to (\ref{C2}) with $B_0=0$ can be given by the smallest integer $m$ such that the diagonal prolongations  to $(\mathcal{M}^\times_4)^m\simeq  \left(\mathbb{R}^{4}_0 \right)^m$ of $X_1,X_2,X_3$  are linearly independent at a generic point \cite{CGM07}. The coordinate expressions for $X_1,X_2,X_3$, given in \eqref{2levelXj}, show that they are already linearly independent at a generic point of  $\mathcal{M}^\times_4$. Hence, the superposition rule does depend on a mere particular solution, which is better than the standard quantum linear superposition rule for the linear system (\ref{C2}), which depends on two particular solutions.
		
		As mentioned in Section \ref{LSLS}, the superposition rule for $X$ can be obtained from certain first-integrals for the diagonal prolongations $X_1^{[2]},X_2^{[2]},X_3^{[2]}$ of $X_1,X_2,X_3$ to $(\mathcal{M}_4^\times)^2\simeq(\mathbb{R}^{4}_0)^{2}$. Using the definition of diagonal prolongations of vector fields and sections of vector bundles given in Section \ref{LSLS}, we can prove that, as the Lie derivative of $g,\omega,J$ with respect to any $X \in V$ is zero, the same happens for the diagonal prolongations $g^{[2]},\omega^{[2]},J^{[2]}$ relative to any $X^{[2]} \in V^{[2]}:= \langle X_1^{[2]},X_2^{[2]},X_3^{[2]} \rangle$, where
		$$
		\begin{gathered}
		\omega^{[2]}=\sum_{r=0}^1 \sum_{j=1}^2{\rm d}{q^{(r)}_j}\wedge {\rm d}{p^{(r)}_j},\qquad\quad 
		g^{[2]} = \sum_{r=0}^1 \sum_{j=1}^2 ({\rm d}{q^{(r)}_j}\otimes {\rm d}{q^{(r)}_j} + {\rm d}{p^{(r)}_j} \otimes {\rm d}{p^{(r)}_j}), \\ 
		J^{[2]}=\sum_{r=0}^1 \sum_{j=1}^2 \left(\frac{\partial}{\partial p^{(r)}_j}\otimes\d q^{(r)}_j - \frac{\partial}{\partial q^{(r)}_j}\otimes \d p^{(r)}_j \right).
		\end{gathered}
		$$
		Additionally, the vector fields $\Delta^{(0)}$, $\Delta^{(1)}$, $\Gamma^{(0)}$ and $\Gamma^{(1)}$, namely the vector fields $\Delta$ and $\Gamma$ defined on each copy of $\mathcal{M}^\times_4$ within $(\mathcal{M}^\times_4)^2$,  commute with $X_1^{[2]},X_2^{[2]},X_3^{[2]}$. The tensor field $S_{01}$, defined in \eqref{eq:Srs} remains invariant under the evolution, namely $\mathcal{L}_{X^{[2]}}S_{01}=0$ for any $X^{[2]} \in V^{[2]}$.
		
		To obtain the superposition rule for $X$, four common first-integrals $I^c_1, I^s_1, I^c_2$,
		and $ I^s_2$ for $X_1^{[2]},X_2^{[2]},X_3^{[2]}$ are needed. Additionally, we must demand 
		\begin{equation}\label{Con2}
		\det \left( \frac{\partial(I^c_1, I^s_1, I^c_2, I^s_2)}{\partial(q^{(0)}_1,p^{(0)}_1,q^{(0)}_2,p^{(0)}_2)} \right)\neq 0.
		\end{equation}
		Some common first-integrals for $X_1^{[2]},X_2^{[2]}$, and $X_3^{[2]}$ can be obtained geometrically from the invariance with respect to such vector fields of several geometric structures previously described:
		\begin{equation}\label{EQU}
		\begin{gathered}
		{g}^{[2]} ( \Delta^{(0)}, \Delta^{(0)}) = {g}^{[2]} ( \Gamma^{(0)}, \Gamma^{(0)}), \quad
		{g}^{[2]} ( \Delta^{(1)}, \Delta^{(1)}) = {g}^{[2]} ( \Gamma^{(1)}, \Gamma^{(1)}), \\
		{g}^{[2]} ( \Delta^{(0)}, S_{01} \Delta^{(1)} ) = {g}^{[2]} ( S_{10} \Delta^{(0)}, \Delta^{(1)} ), \quad
		{\omega}^{[2]} ( \Delta^{(0)}, S_{01} \Delta^{(1)} ) =  {g}^{[2]} ({J}^{[2]} \Delta^{(0)}, S_{01} \Delta^{(1)} ), \quad \mbox{etc.}
		\end{gathered}
		\end{equation}
		A simple but long calculation shows that we cannot construct among (\ref{EQU}) four functions satisfying (\ref{Con2}).
		
		Another first-integral for $X_1^{[2]},X_2^{[2]}$, and $X_3^{[2]}$ can be obtained from the fact that the complex volume element on $\mathcal{M}_4^\times$, understood as a complex manifold $\mathbb{C}^2_0$, reads
		$$
		\Omega = {\rm d}z_1 \wedge {\rm d}z_2 = ({\rm d}q_1\wedge {\rm d}q_2 - {\rm d}p_1 \wedge {\rm d}p_2) + {\rm i} ({\rm d}q_1 \wedge {\rm d}p_2 + {\rm d}p_1 \wedge {\rm d}q_2).
		$$
		This volume element defines two real closed non-degenerate 2-forms $\Omega_R$, $\Omega_I$ on $\mathcal{M}_4^\times$:
		\begin{equation*}
		\Omega_R := {\rm d}q_1\wedge {\rm d}q_2 - {\rm d}p_1 \wedge {\rm d}p_2, \quad
		\Omega_I := {\rm d}q_1 \wedge {\rm d}p_2 + {\rm d}p_1 \wedge {\rm d}q_2, 
		\end{equation*}
		which satisfy that $\mathcal{L}_{Y}\Omega_R=\mathcal{L}_{Y}\Omega_I=0$ for any $Y \in V$.
		The diagonal prolongations of $\Omega_R$ and $\Omega_I$ to $(\mathcal{M}_4^\times)^2$ allow us to obtain new first-integrals for $X_1^{[2]},X_2^{[2]},X_3^{[2]}$: 
		$$
		\Omega_R^{[2]} (\Delta^{(0)}, S_{01} \Delta^{(1)}) = -\Omega_R^{[2]} (\Gamma^{(0)}, S_{01} \Gamma^{(1)}), \qquad
		\Omega_I^{[2]} (\Delta^{(0)}, S_{01} \Delta^{(1)}) = -\Omega_I^{[2]} (\Gamma^{(0)}, S_{01} \Gamma^{(1)}).
		$$From the set of first-integrals on $(\mathcal{M}_4^\times)^2$ so obtained, let us choose four of them as follows:
		{\small 
		\begin{align*}
		I^c_1 (\psi^{(0)}, \psi^{(1)}) := & \frac{{g}^{[2]} ( \Delta^{(0)}, S_{01} \Delta^{(1)} )}{{g}^{[2]} ( \Delta^{(1)}, \Delta^{(1)} )} = \frac{{g}^{[2]} ( \Gamma^{(0)}, S_{01} \Gamma^{(1)} )}{{g}^{[2]} ( \Gamma^{(1)}, \Gamma^{(1)} )} = \frac{\sum_{j=1}^2(q^{(0)}_jq^{(1)}_j+p^{(0)}_jp^{(1)}_j)}{\sum_{j=1}^2[(q^{(1)}_j)^2+(p^{(1)}_j)^2]}, \\
		I^s_1 (\psi^{(0)}, \psi^{(1)}) := & \frac{{\omega}^{[2]} ( \Delta^{(0)}, S_{01} \Delta^{(1)} )}{{g}^{[2]} ( \Delta^{(1)}, \Delta^{(1)} )} = \frac{{\omega}^{[2]} ( \Gamma^{(0)}, S_{01} \Gamma^{(1)} )}{{g}^{[2]} ( \Gamma^{(1)}, \Gamma^{(1)} )} = \frac{\sum_{j=1}^2(q^{(0)}_jp^{(1)}_j - p^{(0)}_jq^{(1)}_j)}{\sum_{j=1}^2[(q^{(1)}_j)^2+(p^{(1)}_j)^2]}, \\
		I^c_2 (\psi^{(0)}, \psi^{(1)}) := & \frac{\Omega_R^{[2]} (\Delta^{(0)}, S_{01} \Delta^{(1)})}{{g}^{[2]} ( \Delta^{(1)}, \Delta^{(1)} )} = -\frac{\Omega_R^{[2]} (\Gamma^{(0)}, S_{01} \Gamma^{(1)})}{{g}^{[2]} ( \Gamma^{(1)}, \Gamma^{(1)} )} = \frac{q^{(0)}_1 q^{(1)}_2 - p^{(0)}_1 p^{(1)}_2 - q^{(0)}_2 q^{(1)}_1 + p^{(0)}_2 p^{(1)}_1}{\sum_{j=1}^2[(q^{(1)}_j)^2+(p^{(1)}_j)^2]}, \\
		I^s_2 (\psi^{(0)}, \psi^{(1)}) := & \frac{\Omega_j^{[2]} (\Delta^{(0)}, S_{01} \Delta^{(1)})}{{g}^{[2]} ( \Delta^{(1)}, \Delta^{(1)} )} = -\frac{\Omega_j^{[2]} (\Gamma^{(0)}, S_{01} \Gamma^{(1)})}{{g}^{[2]} ( \Gamma^{(1)}, \Gamma^{(1)} )} = \frac{q^{(0)}_1 p^{(1)}_2 + p^{(0)}_1 q^{(1)}_2 - q^{(0)}_2 p^{(1)}_1 - p^{(0)}_2 q^{(1)}_1}{\sum_{j=1}^2[(q^{(1)}_j)^2+(p^{(1)}_j)^2]}.
		\end{align*}}
		The normalization factors allow us to obtain a simple superposition rule. These functions satisfy that
		$$
		\det \left( \frac{\partial( I^c_1, I^s_1, I^c_2, I^s_2)}{\partial (q_1^{(0)},p_1^{(0)},q_2^{(0)},p_2^{(0)})} \right)=\left( (q^{(1)}_1)^2 + (p^{(1)}_1)^2 + (q^{(1)}_2)^2 + (p^{(1)}_2)^2 \right)^{-2} \neq 0.
		$$
		The matrix of partial derivatives is non-singular for any point in $\mathcal{M}_4^\times$. Therefore, the system of equations
		$$
		I^c_1 (\psi^{(0)}, \psi^{(1)}) = k_1, \quad
		I^s_1 (\psi^{(0)}, \psi^{(1)}) = k_2, \quad
		I^c_2 (\psi^{(0)}, \psi^{(1)}) = k_3, \quad
		I^s_2 (\psi^{(0)}, \psi^{(1)}) = k_4,
		$$
		can be solved for $\psi^{(0)}:=(q^{(0)}_1, p^{(0)}_1, q^{(0)}_2, p^{(0)}_2)$, giving rise to the superposition rule 
		\begin{equation}
		\label{supRuleH}
		\Phi:(\psi^{(1)},k)\in \mathcal{M}^\times_4\times \mathcal{M}^\times_4\mapsto \psi^{(0)}:=A(k)\psi^{(1)}\in \mathcal{M}^\times_4, \quad
		A(k):=\left[\begin{array}{cccc}
		k_1& - k_2& k_3 &k_4\\
		k_2& k_1 & k_4&- k_3\\
		-k_3& -k_4 & k_1&- k_2\\
		-k_4& k_3 & k_2&k_1\\
		\end{array}\right],
		\end{equation}
		with $k = (k_1, k_2, k_3, k_4)$.

		\subsection{The superposition rule for the Lie system on $\mathcal{S}_2$}
		
		The projection $\pi_\mathcal{MS}: \mathcal{M}_4^\times\rightarrow  \mathcal{S}_2$ imposes an equivalence relation between points in $\mathcal{M}_4^\times$ that differ only on a positive real multiplicative constant. Therefore, each equivalence class of a point $\psi\in \mathcal{M}_4^\times$ is of the form $[\psi]_\mathcal{S} :=\{\lambda\psi \mid \lambda>0\}$. Hence, each equivalence class can be represented by its unique intersection with the unit sphere and $\mathcal{M}_4^\times/\mathbb{R}_{+} \simeq \mathcal{S}_2$. We will consider the natural embedding $\iota_\mathcal{S}: \mathcal{S}_2 \hookrightarrow \mathcal{M}_4$. The pullback of this embedding defines a presymplectic structure $\iota_\mathcal{S}^* \omega$ and a Riemannian metric $\iota_\mathcal{S}^* g$ on $\mathcal{S}_2$. 
		
		As (\ref{C2}) is an $\mathbb{R}$-linear system over $\mathbb{R}$, then the Lie system $X_{\mathcal{M}^\times_{4}}$ can be projected through $\pi_\mathcal{MS}$ onto a system $X_{\mathcal{S}_2}$ on $\mathcal{S}_2$ which, indeed, is the restriction of (\ref{C2}) to $\mathcal{S}_2$, as in Proposition \ref{ResSch}. This proposition also ensures that the vector fields $X_1|_{\mathcal{S}_2}, X_2|_{\mathcal{S}_2}$, and $X_3|_{\mathcal{S}_2}$ are Hamiltonian with respect to the presymplectic structure $\iota_\mathcal{S}^* \omega$, admitting Hamiltonian functions $\bar h_i:=\iota_\mathcal{S}^*h_i$. Also, since the vector fields $X_1,X_2$, and $X_3$ are Killing vector fields for $g$ on $\mathcal{M}_4^\times$, then the vector fields $X_1|_{\mathcal{S}_2}, X_2|_{\mathcal{S}_2}$, and $X_3|_{\mathcal{S}_2}$ are also Killing vector fields with respect to $\iota_\mathcal{S}^* g$.

		Corollary \ref{Cor1} ensures that the restrictions  $X_1|_{\mathcal{S}_2}, X_2|_{\mathcal{S}_2}, X_3|_{\mathcal{S}_2}$ are linearly independent at a generic point of $\mathcal{S}_2$, then $X|_{\mathcal{S}_2}$ admits a superposition rule depending on a unique particular solution. This superposition rule can be obtained by using a similar approach as in the above section, i.e. obtaining three common first-integrals for the diagonal prolongations $X^{[2]}_1|_{\mathcal{S}_2}, X^{[2]}_2|_{\mathcal{S}_2}, X^{[2]}_3|_{\mathcal{S}_2}$, which are Killing vector fields with respect to $(\iota_\mathcal{S}^*g)^{[2]}$ and Hamiltonian vector fields relative to the presymplectic structure $(\iota_\mathcal{S}^*\omega)^{[2]}$. The latter can be employed to obtain the common first-integrals through invariant functions constructed through $(\iota_\mathcal{S}^*g)^{[2]}$, $(\iota_\mathcal{S}^*\omega)^{[2]}$. Importantly, the vector $\Delta$ is not tangent to $\mathcal{S}_2$ and it cannot be used to construct invariants. 
		More easily, we can obtain the pullback via $\iota_\mathcal{S}$ of the first integrals on $(\mathcal{M}_4^\times)^2$ computed in the above section, which allows to determine the superposition rule.
		
		Instead of the above, we will use the following approach, which allows us to obtain the superposition rule for $X_{\mathcal{S}_2}$ from of the superposition for $X_{\mathcal{M}_4^\times}$.
		Observe that $X_{\mathcal{M}_4^\times}$ is a Lie system on $\mathcal{M}_4^\times$ with a superposition rule $\Phi:\mathcal{M}_4^\times \times \mathcal{M}_4^\times\rightarrow \mathcal{M}_4^\times$ and that $(X_{\mathcal{M}_4^\times})_t$ is tangent to a submanifold $\mathcal{S}_2\subset \mathcal{M}_4^\times$ for each $t\in\mathbb{R}$. Assume also that there exists $\bar S \subset \mathcal{M}_4^\times$ such that $\Phi(\mathcal{S}_2 \times \bar S)= \mathcal{S}_2$. Then, the initial superposition rule can be restricted to elements on $\mathcal{S}_2$ giving rise to a new superposition principle.

		Indeed, let us consider the superposition rule $\Phi$ defined above and evaluated on points $\psi_S^{(1)}, k_S \in \mathcal{S}_2$, i.e. $\|\psi_S^{(1)}\| = \| k_S \| = 1$. The resulting point $\Phi (\psi_S^{(1)}, k_S)$ satisfies that
		$$
		\| \Phi (\psi_S^{(1)}, k_S) \| = |\det A(k_S)| \| \psi_S^{(1)} \| = \| k_S \|^4 \| \psi_S^{(1)} \| = 1 \Rightarrow \Phi (\psi_S^{(1)}, k_S) \in \mathcal{S}_2.
		$$
		Conversely, there always exists, for  points $\psi_S^{(0)}\in\mathcal{S}_2$  and $k_S\in\mathcal{S}_2$, a point $\psi_S^{(1)}\in \mathcal{S}_2$ such that $\Phi(\psi_S^{(1)},k_s)=\psi_S^{(0)}$.
		Hence $X_{\mathcal{S}_2}$ admits a superposition rule
		\begin{equation}
		\Phi_S: (\psi_S^{(1)}, k_S) \in \mathcal{S}_2 \times \mathcal{S}_2 \mapsto A(k_S) \psi_S^{(1)} \in \mathcal{S}_2,
		\end{equation}
		with $A(k)$ given by \eqref{supRuleH}.

		\subsection{Superposition rules on $ \mathcal{R}_2$ and $\mathcal{P}_2$}

		Let us obtain a superposition rule for the system (\ref{YField}) on $\mathcal{R}_2$ and prove that it depends on two particular solutions.

		We employ the global coordinate system $\{x , y, z\}$ on $\mathcal{R}_2$ suggested in Lemma \ref{CoorHU} and presented in Example \ref{Coor}. To simplify the notation, ${\rm\bf x}:=(x,y,z)\neq 0$ will represent an arbitrary point of $\mathcal{R}_2\simeq\mathbb{R}^3_0$.
			
			Recall that the Poisson structure $\Lambda$ on $\mathcal{M}_4^\times$ can be projected onto $\mathcal{R}_2$ giving rise to a contravariant tensor field $\widehat{\Lambda} = \pi_{\mathcal{MR} *} \Lambda$ on $\mathcal{R}_2$, presented in \eqref{eq:2LambdaHat}, which is a Poisson tensor field on $\mathcal{R}_2$.

		To obtain a superposition rule for (\ref{YField}), we have to  find the smallest $m\in\mathbb{N}$ so that $Y_1^{[m]}, Y_2^{[m]}, Y_3^{[m]}$ are linearly independent at a generic point of $\mathcal{R}_2^m$. 
		Since  $Y_1,Y_2$, and $Y_3$ span a two-dimensional distribution and their are linearly independent over $\mathbb{R}$ (see Example \ref{ExHU}), then their diagonal prolongations to $\mathcal{R}_2$ span, at least, a distribution of rank three and $Y_1^{[2]}, Y_2^{[2]}$, and $ Y_3^{[2]}$ are linearly independent at a generic point \cite{CGM07,Dissertations}. 
		Hence, a superposition rule for (\ref{YField}) results from giving three functionally independent common first-integrals  $I_1,I_2,I_3:\mathcal{R}_2^3\rightarrow \mathbb{R}$ for the diagonal prolongations

		\begin{align*}
		{Y}^{[3]}_1 & = z^{(0)} \frac{\partial}{\partial y^{(0)}} - y^{(0)} \frac{\partial}{\partial z^{(0)}} + z^{(1)} \frac{\partial}{\partial y^{(1)}} - y^{(1)} \frac{\partial}{\partial z^{(1)}} + z^{(2)} \frac{\partial}{\partial y^{(2)}} - y^{(2)} \frac{\partial}{\partial z^{(2)}}, \\
		{Y}^{[3]}_2 & =  x^{(0)} \frac{\partial}{\partial z^{(0)}} - z^{(0)} \frac{\partial}{\partial x^{(0)}} + x^{(1)} \frac{\partial}{\partial z^{(1)}} - z^{(1)} \frac{\partial}{\partial x^{(1)}} + x^{(2)} \frac{\partial}{\partial z^{(2)}} - z^{(2)} \frac{\partial}{\partial x^{(2)}}, \\
		{Y}^{[3]}_3 & = y^{(0)} \frac{\partial}{\partial x^{(0)}} - x^{(0)} \frac{\partial}{\partial y^{(0)}} + y^{(1)} \frac{\partial}{\partial x^{(1)}} - x^{(1)} \frac{\partial}{\partial y^{(1)}} + y^{(2)} \frac{\partial}{\partial x^{(2)}} - x^{(2)} \frac{\partial}{\partial y^{(2)}}.
		\end{align*}
		satisfying $\det(\partial (I_1,I_2,I_3)/\partial (x^{(0)}, y^{(0)}, z^{(0)}))\neq 0$.
		
		The needed first-integrals can be obtained from the diagonal prolongations of the Lie symmetry
		$$
		\Delta= x\frac{\partial}{\partial x}+y\frac{\partial}{\partial y}+z\frac{\partial}{\partial z}.
		$$
		Firstly, let us define the normalization function
		\begin{equation}
		N ({\bf x}) := \widehat{g} (\Delta, \Delta) = (x^2 + y^2 + z^2)^{1/2},\qquad {\bf x}\in \mathcal{R}_2,
		\end{equation}
		with $\widehat{g}$ given by \eqref{eq:2gHat}. This function is a first integral of $X_{\mathcal{R}_2}$ and hence it can be understood as a constant of motion of any of its prolongations. Consider now the invariant functions on $\mathcal{R}_2^3$ of the form
		\begin{align*}
		I_i ({\bf x}^{(0)}, {\bf x}^{(1)}, {\bf x}^{(2)}) & := N({\bf x}^{(0)}) g^{[3]} ( \Delta^{(0)}, S_{0i} \Delta^{(1)} )=x^{(0)} x^{(j)}+y^{(0)} y^{(j)}+z^{(0)} z^{(j)}, \quad j = 1,2, \\
		I_3 ({\bf x}^{(0)},{\bf x}^{(1)},{\bf x}^{(2)}) & := N({\bf x}^{(0)}) g^{[3]} ( \Delta^{(0)}, \Delta^{(0)} )= (x^{(0)})^2+(y^{(0)})^2+(z^{(0)})^2,
		\end{align*}
		Since
		\begin{equation}\label{conp}
		\det \left( \frac{\partial (I_1,I_2,I_3)}{\partial(x^{(0)},y^{(0)},z^{(0)})} \right) = \det\left[ \begin{array}{ccc}
		x^{(1)} & y^{(1)} &z^{(1)} \\
		x^{(2)} & y^{(2)} & z^{(2)}\\
		2x^{(0)} & 2y^{(0)} & 2z^{(0)}
		\end{array}\right] \neq 0
		\end{equation}
		at a generic point $({\bf x}^{(0)},{\bf x}^{(1)},{\bf x}^{(2)})\in \mathcal{R}_2^3$, obtaining the superposition rule is equivalent to solving for ${\bf x}$ the following system of equations in $\mathcal{R}_2^2=\mathbb{R}_0^3$:
		\begin{equation}
		\label{eq:sistR3}
		{\bf x} \cdot {\bf x}_1 = k_1, \qquad
		{\bf x} \cdot {\bf x}_2 = k_2, \qquad
		{\bf x} \cdot {\bf x} = k_3,
		\end{equation}
		for some $k_1, k_2, k_3 \in \mathbb{R}$, with $k_3 >0$. Since ${\bf x}_1$ and ${\bf x}_2$ are not collinear when (\ref{conp}) holds, the above system is easily solved in ${\bf x}$ by defining an orthonormal system relative to the standard scalar product on $\mathcal{R}_2\simeq\mathbb{R}^3$:
		$$
		{\bf x}'_1 := \frac{{\bf x}_1}{\lVert {\bf x}_1 \rVert} , \qquad
		{\bf x}'_2 := \frac{\lVert {\bf x}_1 \rVert^2 {\bf x}_2 - ({\bf x}_1 \cdot {\bf x}_2) {\bf x}_1}{\lVert {\bf x}_1 \rVert \sqrt{\lVert {\bf x}_1 \rVert^2 \lVert {\bf x}_2 \rVert^2 - ({\bf x}_1 \cdot {\bf x}_2)^2}}.
		$$
		These two new vectors together with their cross product, ${\bf x}'_1 \times {\bf x}'_2$, conform an orthonormal basis for $\mathbb{R}^3$. From \eqref{eq:sistR3}, the general expression for $\bf x$ is
		\begin{equation}\label{solXNorm}
		{\bf x} = k'_1 {\bf x}'_1 + k'_2 {\bf x}'_2 \pm \sqrt{k_3 - (k'_1)^2 - (k'_2)^2} ( {\bf x}'_1 \times {\bf x}'_2),
		\end{equation}
		where the coefficients $k'_1$ and $k'_2$ are
		$$
		k'_1 = {\bf x} \cdot {\bf x}'_1 = \frac{k_1}{\lVert {\bf x}_1 \rVert}, \qquad
		k'_2 = {\bf x} \cdot {\bf x}'_2 = \frac{k_2 \lVert {\bf x}_1 \rVert^2 - k_1 ({\bf x}_1 \cdot {\bf x}_2)}{\lVert {\bf x}_1 \rVert \sqrt{\lVert {\bf x}_1 \rVert^2 \lVert {\bf x}_2 \rVert^2 - ({\bf x}_1 \cdot {\bf x}_2)^2}}.
		$$
		Replacing $k_1'$ and $k_2'$ in \eqref{solXNorm}, the solution to the system of equations (\ref{eq:sistR3}) is
		\begin{align*}
		{\bf x} = & \frac{\delta_{12}{\bf x}_1 +\delta_{21}{\bf x}_2\pm\sqrt{k_3 [\lVert {\bf x}_1 \rVert^2 \lVert {\bf x}_2 \rVert^2 -  ({\bf x}_1 \cdot {\bf x}_2)^2 ]- (k_1{\bf x}_1 -k_2 {\bf x}_2 )^2} {\bf x}_1 \times {\bf x}_2    }{\lVert {\bf x}_1 \rVert^2 \lVert {\bf x}_2 \rVert^2 - ({\bf x}_1 \cdot {\bf x}_2)^2},
		\end{align*}
		where $\delta_{lj}:=k_l\|{\bf x}_j\|^2-k_j({\bf x}_l\cdot {\bf x}_j)$. 
		As the Lie system $X_{\mathcal{R}_2}$ is linear in the chosen coordinate system and the Riemannian metric related to the standard scalar product on $\mathcal{R}_2\simeq\mathbb{R}^3_0$ is invariant under the elements of $V_{\mathcal{R}_2}$, it follows that $\|{\bf x}_1\|^2$, $\|{\bf x}_2\|^2$
		and  ${\bf x}_1\cdot{\bf x}_2$ are constant along particular solutions of $X_{\mathcal{R}_2}$. Then, the above expression gives rise to a superposition rule $\Phi:({\bf x_1},{\bf x_2},(k_1,k_2,k_3))\in \mathcal{R}^2_2\times \mathcal{A}\mapsto {\bf x}\in \mathcal{R}_2$, with $\mathcal{A}=:\{(k_1,k_2,k_3:k_3\neq 0)\}$, of the form
		\begin{align}
		\label{supRulePGamma}
		{\bf x} = \delta_{12}{\bf x}_1 +\delta_{21}{\bf x}_2+{\rm sign}(k_3) \sqrt{|k_3|k_{12}- (k_1{\bf x}_1 -k_2 {\bf x}_2 )^2} ({\bf x}_1 \times {\bf x}_2) ,
		\end{align}
		where $k_{12}:=\lVert {\bf x}_1 \rVert^2 \lVert {\bf x}_2 \rVert^2 -  ({\bf x}_1 \cdot {\bf x}_2)^2$.
		
		In view of Proposition \ref{RnPn}, deriving a superposition rule for $X_{\mathcal{P}_2}$ amounts to obtain a superposition rule for the solutions to $X_{\mathcal{R}_2}$ on $\pi_\mathcal{SR} (\mathcal{S}_2)$, namely those equivalence classes of $\mathcal{R}_2$ coming from elements of $\mathcal{M}_4^\times$ with the same module. To obtain the superposition rule for the system $X_{\mathcal{P}_2}$ on $\P_2$, consider the natural embedding of $\mathcal{P}_2$ into $\mathcal{R}_2$ whose image is the set of elements $(x,y,z) \in \mathcal{R}_2$ such that $x^2 + y^2 + z^2 = \langle \psi,\psi\rangle^2/16= 1$. Therefore, $\P_2$ is diffeomorphic to a sphere $\mathcal{S}^2 \subset \mathcal{R}^2\simeq \mathbb{R}^3_0$. Consider the superposition rule defined for $X_{\mathcal{R}_2}$ when restricted to points in $S^2$, i.e. with $\| {\bf x}_1 \| = \| {\bf x}_2 \| = 1$. The set of constants has to be constrained in order to obtain solutions in $\mathcal{S}^2$. From \eqref{eq:sistR3}, the constraints are $|k_1|, |k_2| \leq 1$, $k_3 =1$. In consequence, the superposition rule for $\mathcal{P}_2\simeq \mathcal{S}^2$ is 
		\begin{equation}
		{\bf x} =\frac{ (k_1-k_2 {\bf x}_1\cdot{\bf x}_2){\bf x}_1+ (k_2-k_1 {\bf x}_1\cdot{\bf x}_2){\bf x}_2 \pm \sqrt{1 - ({\bf x}_1 \cdot {\bf x}_2 )^2- (k_1{\bf x}_1 -k_2 {\bf x}_2 )^2} ({\bf x}_1 \times {\bf x}_2)}{1 - ({\bf x}_1 \cdot {\bf x}_2)^2} , 
		\end{equation}
		with ${\bf x}_1, {\bf x}_2 \in \mathcal{P}_2$ and $|k_1|, |k_2| \leq 1.$

		When ${\bf x}_1$ and ${\bf x}_2$ are replaced by two generic particular solutions of the system within $\mathcal{S}^2\subset \mathcal{R}_2\simeq \mathbb{R}^3_0$, the general solution is obtained.

		\section{Superposition rules for $n$-levels systems on $\mathcal{M}_{2n}^\times$ and $\mathcal{S}_n$}

		Let us obtain a superposition rule for $X_{\mathcal{M}_{2n}^\times}$, where we assume $X_1,\ldots,X_{n^2-1}$ to belong to the Lie algebra of fundamental vector fields for the action of $SU(n)$ on $\mathcal{M}_{2n}^\times$. Our aim is to illustrate our previous theory while obtaining the explicit expression for a superposition rule depending just on $n-1$ particular solutions.
		
		Theorem \ref{MT1} shows that the smallest $m$ turning linearly independent the diagonal prolongations of the vector fields $X_\alpha$ is given by $m=n-1$. Hence, the superposition rule does depend on $n-1$ particular solutions, which is better than the standard linear superposition rule for the $t$-dependent Schr\"odinger equation depending on $n$ particular solutions. In contrast with the two-level system, the superposition rule depending on $n-1$ particular solutions for $n>2$ is not linear.
		
		The superposition rule can be derived through $2n$ common first-integrals $I^c_1, I^s_1,\ldots ,I^c_n, I^s_n:(\mathcal{M}_{2n}^\times)^{n}  \rightarrow \mathbb{R}$  for all the diagonal prolongations $X^{[n]}_\alpha$ on $(\mathcal{M}_{2n}^\times)^{n}$, with $\alpha=\overline{1,n^2-1}$. Additionally, these first-integrals give rise to a superposition rule provided that 
		$$\det (\partial(I^c_1, I^s_1,\ldots ,I^c_n, I^s_n)/\partial(q_1^{(0)},p_1^{(0)} ,\ldots, q_{n}^{(0)} ,p_{n}^{(0)}))\neq 0.$$ 
		
		Let us obtain the first integrals geometrically. Since the $X_\alpha$ are K\"ahler vector fields 
		elative to the K\"ahler structure $(g,\omega,J)$ on $\mathcal{M}^\times_{2n}$, their diagonal prolongations 
		$X^{[n]}_\alpha$ are K\"ahler relative to the diagonal prolongation $(g^{[n]},\omega^{[n]},J^{[n]})$  to $(\mathcal{M}_{2n}^\times)^{n}$ of the K\"ahler structure $(g,\omega,J)$, namely
		$$
		\omega^{[n]} = \sum_{j=1}^{n} \sum_{a=0}^{n-1} {\rm d}{q^{(a)}_j}\wedge {\rm d} {p^{(a)}_j},\,\,
		g^{[n]} = \sum_{j=1}^{n} \sum_{a=0}^{n-1} ({\rm d}{q^{(a)}_j} \otimes {\rm d}{q^{(a)}_j} + {\rm d}{p^{(a)}_j} \otimes {\rm d}{p^{(a)}_j}),
		$$
		$$
		J^{[n]} = \sum_{j=1}^{n} \sum_{a=0}^{n-1} \frac{\partial}{\partial p^{(a)}_j} \otimes {\rm d}{q^{(a)}_j} - \frac{\partial}{\partial q^{(a)}_j} \otimes {\rm d}{p^{(a)}_j}.
		$$
		Similarly, if $X$ is a Hamiltonian vector field relative to $\omega$ with Hamiltonian function $h_X$, then $X^{[n]}$ is a Hamiltonian vector field with Hamiltonian function $h^{[n]}_X$. 
		
		As the vector fields $X^{[n]}_\alpha$ are Killing vector fields with respect to $g^{[n]}$ and symmetries of the tensor fields $S_{rs}$ for $r,s=0,\overline{1,n-1}$ and $r\neq s$, presented in \eqref{eq:SrsDef} and \eqref{eq:Srs}, we can obtain the following common first-integrals for all such vector fields:
		\begin{equation}
		\label{defIk}
		\begin{gathered}
		I^c_k := g^{[n]} (\Delta^{(0)}, S_{0k} (\Delta^{(k)})) = \sum_{j=1}^{n} (q^{(0)}_j q^{(k)}_j + p^{(0)}_j p^{(k)}_j), \qquad k = 1, \ldots, n-1, \\
		I^s_k := g^{[n]} (\Gamma^{(0)}, S_{0k}(\Delta^{(k)})) = \sum_{j=1}^{n} (q^{(0)}_j p^{(k)}_j - p^{(0)}_j q^{(k)}_j), \qquad k = 1, \ldots, n-1.
		\end{gathered}
		\end{equation}
		These functions satisfy that
		\begin{equation}
		J^{[n]} ({\rm d}I^c_k) = {\rm d}I^s_k, \qquad k = 1, \ldots, n-1.
		\end{equation}
		For a given value of $k$, the functions $I^c_k$ and $I^s_k$ are functionally independent, while functions with different values of $k$ involve different variables. Hence all of these functions are functionally independent among them.
		
		Observe that the functions in \eqref{defIk} are first-integrals not only for $X_\alpha^{[n]}$, but also for $\Gamma^{[n]}$. We can also obtain functions which are not first-integrals of $\Gamma^{[n]}$ with help of the $n$-forms $\Omega_R$ and $\Omega_I$. Let $I^c_{n}$, $I^s_{n}$ be the function defined as
		\begin{equation}
		\begin{gathered}
		I^c_{n} := \Omega_R^{[n]}(\Delta^{(0)}, S_{01} (\Delta^{(1)}), \ldots, S_{0(n-1)} (\Delta^{(n-1)})) =\mathfrak{Re} (\det(\psi^{(0)},\ldots ,\psi^{(n-1)})), \\
		I^s_{n} := \Omega_I^{[n]}(\Delta^{(0)}, S_{01} (\Delta^{(1)}), \ldots, S_{0(n-1)} (\Delta^{(n-1)})) =\mathfrak{Im} (\det(\psi^{(0)} , \ldots ,\psi^{(n-1)})).
		\end{gathered}
		\end{equation}
		These functions satisfy that $J^{[m]} ({\rm d}I^c_n) = {\rm d}I^s_n$, so they are functionally independent among themselves. As they are not first-integrals of $\Gamma^{[n]}$, they are functionally independent of functions in (\ref{defIk}).
		
		To sum up, there exist $2n$ first-integrals of the action of $\mathfrak{su}(n)$ on $\mathcal{M}_{2n}^\times$ given by $I^c_1, I^s_1$, ..., $I^c_n, I^s_n$. They satisfy that the matrix of partial derivatives of these functions with respect to the coordinates of $\psi^{(0)}$ is non-singular. Therefore, the solution $\psi^{(0)}$ to the equations 
		$$
		I^c_j (\psi^{(0)}, \psi^{(1)}, \ldots \psi^{(n-1)}) = k_{2j-1}, \quad
		I^s_j (\psi^{(0)}, \psi^{(1)}, \ldots \psi^{(n-1)}) = k_{2j}, \quad
		j=\overline{1,n},
		$$
		can be obtained, at least locally, in terms of the coordinates of $\psi^{(1)}$, ..., $\psi^{(n-1)}$ and $2n$ real constants $k_1$, ..., $k_{2n}$.
		
		Since all functions are linear in the coordinates of $\psi^{(0)}$, then $\psi^{(0)}$, and consequently the superposition rule, can be obtained by solving the system
		{\footnotesize
		$$
		\begin{pmatrix}
		\frac{\partial I^c_1}{\partial q^{(0)}_1} & \frac{\partial I^c_1}{\partial p^{(0)}_1} & \cdots & \frac{\partial I^c_1}{\partial q^{(0)}_n} & \frac{\partial I^c_1}{\partial p^{(0)}_n}\\
		\frac{\partial I^s_1}{\partial q^{(0)}_1} & \frac{\partial I^s_1}{\partial p^{(0)}_1} & \cdots & \frac{\partial I^s_1}{\partial q^{(0)}_n} & \frac{\partial I^s_1}{\partial p^{(0)}_n}\\
		\vdots & \vdots & \ddots & \vdots & \vdots\\
		\frac{\partial I^c_n}{\partial q^{(0)}_1} & \frac{\partial I^c_n}{\partial p^{(0)}_1} & \cdots & \frac{\partial I^c_n}{\partial q^{(0)}_n} & \frac{\partial I^c_n}{\partial p^{(0)}_n}\\
		\frac{\partial I^s_n}{\partial q^{(0)}_1} & \frac{\partial I^s_n}{\partial p^{(0)}_1} & \cdots & \frac{\partial I^s_n}{\partial q^{(0)}_n} & \frac{\partial I^s_n}{\partial p^{(0)}_n}
		\end{pmatrix}
		\begin{pmatrix}
		q^{(0)}_1 \\ p^{(0)}_1 \\ \vdots \\ q^{(0)}_n \\ p^{(0)}_n
		\end{pmatrix} =
		\begin{pmatrix}
		q^{(1)}_1 & p^{(1)}_1 & \cdots &  q^{(1)}_n & p^{(1)}_n\\
		p^{(1)}_1 & -q^{(1)}_1 & \cdots &  p^{(1)}_n & -q^{(1)}_n\\
		\vdots & \vdots & \ddots & \vdots & \vdots\\
		q^{(n-1)}_1 & p^{(n-1)}_1 & \cdots &  q^{(n-1)}_n & p^{(n-1)}_n\\
		p^{(n-1)}_1 & -q^{(n-1)}_1 & \cdots &  p^{(n-1)}_n & -q^{(n-1)}_n\\
		\frac{\partial I^c_n}{\partial q^{(0)}_1} & \frac{\partial I^c_n}{\partial p^{(0)}_1} & \cdots & \frac{\partial I^c_n}{\partial q^{(0)}_n} & \frac{\partial I^c_n}{\partial p^{(0)}_n}\\
		\frac{\partial I^s_n}{\partial q^{(0)}_1} & \frac{\partial I^s_n}{\partial p^{(0)}_1} & \cdots & \frac{\partial I^s_n}{\partial q^{(0)}_n} & \frac{\partial I^s_n}{\partial p^{(0)}_n}
		\end{pmatrix}
		\begin{pmatrix}
		q^{(0)}_1 \\ p^{(0)}_1 \\ \vdots \\ q^{(0)}_n \\ p^{(0)}_n
		\end{pmatrix} = 
		\begin{pmatrix} k_1\\ \vdots\\ k_{2n} \end{pmatrix}.
		$$}
		
		Observe that
		\begin{equation}
		\label{eq:9.4}
		\begin{gathered}
		\sum_{j=1}^n (q_j^{(r)} q_j^{(s)} + p_j^{(r)} p_j^{(s)} ) = g^{[n]} (\Delta^{(r)}, S_{rs} (\Delta^{(s)})), \\
		\sum_{j=1}^n (q_j^{(r)} p_j^{(s)} - p_j^{(r)} q_j^{(s)} ) = g^{[n]} (\Gamma^{(r)}, S_{rs} (\Delta^{(s)})),
		\end{gathered} \qquad r \neq s \in\{ 1, \ldots, n-1\}.
		\end{equation}
		These values are constants of motion. Given a set of linearly independent vectors $\psi^{(1)}, \ldots, \psi^{(n-1)} \in \mathcal{M}_{2n}$, one can always find linear combinations of them such that these functions are zero for $r\neq s$. Also, we have that
		\begin{equation*}
		\begin{gathered}
		\sum_{\alpha=1}^n \left[q^{(i)}_\alpha \frac{\partial I^c_n}{\partial q^{(0)}_\alpha} + p^{(i)}_\alpha \frac{\partial I^c_n}{\partial p^{(0)}_\alpha}\right] =
		-\sum_{\alpha=1}^n \left[p^{(i)}_\alpha \frac{\partial I^s_n}{\partial q^{(0)}_\alpha} - q^{(i)}_\alpha \frac{\partial I^s_n}{\partial p^{(0)}_\alpha}\right] 
		= \mathfrak{ Re}(\det(\psi^{(i)}, \psi^{(1)} , \ldots, \psi^{(n-1)}) )=0,\\
		\sum_{\alpha=1}^n \left[p^{(i)}_\alpha \frac{\partial I^c_n}{\partial q^{(0)}_\alpha} - q^{(i)}_\alpha \frac{\partial I^c_n}{\partial p^{(0)}_\alpha}\right] =
		\sum_{\alpha=1}^n \left[q^{(i)}_\alpha \frac{\partial I^s_n}{\partial q^{(0)}_\alpha} + p^{(i)}_\alpha \frac{\partial I^s_n}{\partial p^{(0)}_\alpha}\right]
		=\mathfrak {Im}(\det(\psi^{(i)},\psi^{(1)}, \ldots, \psi^{(n-1)} ))=0,\\
		\sum_{\alpha=1}^n \left[ \left(\frac{\partial I^c_n}{\partial q^{(0)}_\alpha}\right)^2 + \left(\frac{\partial I^c_n}{\partial p^{(0)}_\alpha}\right)^2 \right] = \sum_{\alpha=1}^n \left[ \left(\frac{\partial I^s_n}{\partial q^{(0)}_\alpha}\right)^2 + \left(\frac{\partial I^s_n}{\partial p^{(0)}_\alpha}\right)^2 \right] = \prod_{\alpha=1}^{n-1}\|\psi^{(\alpha)}\|^2,
		\end{gathered}
		\end{equation*}
		for $i\in\overline{1,n-1}$.
		Defining $\Theta:=\prod_{\alpha=1}^{n-1}\|\psi^{(\alpha)}\|^2$ and choosing $\psi^{(1)}, \ldots, \psi^{(n-1)}$, so that the quantities in \eqref{eq:9.4} are zero, we obtain:
		{\footnotesize 
		\begin{equation}
		\begin{pmatrix}
		q^{(0)}_1 \\ p^{(0)}_1 \\ \vdots \\ q^{(0)}_n \\ p^{(0)}_n
		\end{pmatrix} = \begin{pmatrix}
		\dfrac{q^{(1)}_1}{||\psi^{(1)}||^2} & \dfrac{p^{(1)}_1}{||\psi^{(1)}||^2} & \cdots &  \dfrac{q^{(1)}_n}{||\psi^{(1)}||^2} & \dfrac{p^{(1)}_n}{||\psi^{(1)}||^2} \\
		\dfrac{p^{(1)}_1}{||\psi^{(1)}||^2} & -\dfrac{q^{(1)}_1}{||\psi^{(1)}||^2} & \cdots &  \dfrac{p^{(1)}_n}{||\psi^{(1)}||^2} & -\dfrac{q^{(1)}_n}{||\psi^{(1)}||^2} \\
		\vdots & \vdots & \ddots&\vdots&\vdots\\
		\dfrac{q^{(n-1)}_1}{||\psi^{(n-1)}||^2} & \dfrac{p^{(n-1)}_1}{||\psi^{(n-1)}||^2} & \cdots &  \dfrac{q^{(n-1)}_n}{||\psi^{(n-1)}||^2} & \dfrac{p^{(n-1)}_n}{||\psi^{(n-1)}||^2} \\
		\dfrac{p^{(n-1)}_1}{||\psi^{(n-1)}||^2} & -\dfrac{q^{(n-1)}_1}{||\psi^{(n-1)}||^2} & \cdots &  \dfrac{p^{(n-1)}_n}{||\psi^{(n-1)}||^2} & -\dfrac{q^{(n-1)}_n}{||\psi^{(n-1)}||^2} \\ \\
		\dfrac{1}{\Theta} \dfrac{\partial I^c_n}{\partial q^{(0)}_1} & \dfrac{1}{\Theta} \dfrac{\partial I^c_n}{\partial p^{(0)}_1} & \cdots & \dfrac{1}{\Theta} \dfrac{\partial I^c_n}{\partial q^{(0)}_n}  & \dfrac{1}{\Theta} \dfrac{\partial I^c_n}{\partial p^{(0)}_n}\\ \\
		\dfrac{1}{\Theta} \dfrac{\partial I^s_n}{\partial q^{(0)}_1} & \dfrac{1}{\Theta} \dfrac{\partial I^s_n}{\partial p^{(0)}_1} & \cdots & \dfrac{1}{\Theta} \dfrac{\partial I^s_n}{\partial q^{(0)}_n} & \dfrac{1}{\Theta} \dfrac{\partial I^s_n}{\partial p^{(0)}_n}
		\end{pmatrix}^T
		\begin{pmatrix} k_1\\ \vdots\\ k_{2n} \end{pmatrix}.
		\end{equation}}
		
		This expression gives rise to a superposition rule $\Phi:[\mathcal{M}_{2n}^\times]^{n-1}\times \mathcal{M}_{2n}^\times \rightarrow \mathcal{M}_{2n}^\times$ for $X_{\mathcal{M}_{2n}^\times}$.
		Observe that when we choose particular solutions $\psi_1,\ldots\psi_{n-1}$ with norm one, then the matrix of the above system becomes unimodular and hence, when $\sum_{\alpha=1}^{2n}k_\alpha^2=1$, we obtain that $\sum_{\alpha=1}^{n}\left((q^{(0)}_\alpha)^2+(p^{(0)}_\alpha)^2\right)=1$. This allows us to restrict the above superposition rule to a new one of the form $\Phi_S:\mathcal{S}_n^{n-1}\times \mathcal{S}_{n}\rightarrow \mathcal{S}_n$ for $X_{\mathcal{S}_n}$.

		\section{Conclusions and outlook}
		The present work has laid down the basis for the study of quantum systems on finite-dimensional Hilbert spaces and some of its projective spaces through the theory of Lie systems. We have proved that all such systems and their projections onto projective spaces are Lie systems and we have found some of their superposition rules. 
		
		An analogous development can be carried out for Heisenberg equations. We aim to find a formalism based on Lie systems to study those equations in the future. We also recently found that Lie systems  appear in the description of Kossakowski--Lindbland equations, which can be of interest for their analysis. Other topics concerning the geometry of Lie systems in quantum mechanics and their study from differential geometry are also in progress.
		
		\section{Acknowledgements}
		
		J.F. Cari\~nena, J.A. Jover-Galtier and J. Clemente-Gallardo acknowledge partial financial support from MINECO (Spain) grant number MTM2015-64166-C2-1. The research of J. de Lucas was
		supported under the contract 1100/112000/16. Research of J. Clemente--Gallardo and J.A. Jover-Galtier was financed by projects MICINN Grants FIS2013-46159-C3-2-P.
		Research of J.A. Jover-Galtier was financed by DGA grant number B100/13 and by ``Programa de FPU del Ministerio de Educaci\'on, Cultura y Deporte'' grant number FPU13/01587.

\end{document}